\documentclass[a4paper,abstracton]{scrartcl}

\usepackage[utf8]{inputenc}

\usepackage{a4wide}

\usepackage{amsmath}
\usepackage{amsthm}
\usepackage{amssymb}

\usepackage{hyperref}

\usepackage[usenames,dvipsnames]{color}
\usepackage{graphicx}
\usepackage{todonotes}
\usepackage{caption}

\usepackage{tikz}
\usetikzlibrary{automata}
\usetikzlibrary{arrows,shapes}

\usepackage{bbding}

\usepackage{amssymb,amsmath,verbatim,tabularx,enumitem,colonequals,graphicx,psfrag}

\usepackage{cite}

\newcommand{\R}{\mathbb{R}}
\newcommand{\X}{\mathcal{X}}

\usepackage{authblk}
\usepackage[normalem]{ulem} 

\usepackage{framed}
\usepackage{rotating}
\usepackage{multirow}
\usepackage{subfigure}

\newtheorem{theorem}{Theorem}
\newtheorem{lemma}[theorem]{Lemma}

\newcommand{\BIGOP}[1]{\mathop{\mathchoice%
{\raise-0.22em\hbox{\huge $#1$}}%
{\raise-0.05em\hbox{\Large $#1$}}{\hbox{\large $#1$}}{#1}}}

\allowdisplaybreaks

\begin{document}

\title{Approximating multiobjective combinatorial optimization problems with the OWA criterion}

\author[1]{Andr\'e Chassein}
\author[2]{Marc Goerigk}
\author[3]{Adam Kasperski}
\author[4]{Pawe{\l} Zieli\'nski}

\affil[1]{Fachbereich Mathematik, Technische Universit\"at Kaiserslautern, Germany 
(\texttt{chassein@mathematik.uni-kl.de})}
\affil[2]{Department of Management Science, Lancaster University, United Kingdom
 (\texttt{m.goerigk@lancaster.ac.uk})}
\affil[3]{Faculty of Computer Science and Management, Wroc{\l}aw  University of Science and Technology, Poland (\texttt{adam.kasperski@pwr.edu.pl})}
\affil[4]{Faculty of Fundamental Problems of Technology, Wroc{\l}aw  University of Science and Technology, Poland (\texttt{pawel.zielinski@pwr.edu.pl})}

\date{}

\maketitle

\begin{abstract}
The paper  deals with 
a multiobjective combinatorial optimization problem with~$K$ linear cost functions.
 The popular Ordered Weighted Averaging (OWA) criterion is used to aggregate the cost functions and compute a solution. 
 It is well known that minimizing OWA  for most basic combinatorial problems 
 is weakly NP-hard
 even if the number of objectives~$K$ equals two, and strongly NP-hard when $K$ is a part of the input. 
 In this paper, the problem with nonincreasing weights in the OWA criterion   and a  large $K$ is considered.
 A method of reducing the number of objectives by appropriately aggregating  the objective costs before
 solving the problem is proposed.
 It is shown that an optimal solution to the reduced problem
  has a guaranteed worst-case approximation ratio.
  Some new approximation results for the Hurwicz criterion, which is a special case of OWA,
  are also presented.
\end{abstract}

\textbf{Keywords: } multiobjective optimization; ordered weighted averaging; robust optimization; approximation algorithms; combinatorial optimization

\section{Introduction}

In many practical applications of combinatorial optimization, we seek a solution optimizing  more than one objective function (criterion). 
In this case, we typically seek a set of \emph{efficient (Pareto optimal) solutions} 
or 
reduce the multiobjective  optimization problem to  a single objective one, by using  aggregation functions
and finding   \emph{compromise solutions} according to a given aggregation function
(see, e.g.,~\cite{ER05}). 
For surveys on multiobjective combinatorial optimization problems, we refer the reader to~\cite{ER00,ER05, UT94}.  Multiobjective problems arise naturally under uncertainty. If  the cost coefficients are uncertain, then we can provide a sample of various cost realizations, called \emph{scenarios} (states of the world), where
each scenario defines an objective function. 
The aim is to find a solution that has a good performance 
whatever scenario is finally revealed. Robust min-max criteria can be then used to compute such a solution 
(see, e.g.,~\cite{KY97, ABV09, KZ16b}).

In order to aggregate the objective functions, the \emph{Ordered Weighted Averaging} (OWA) criterion, proposed in~\cite{YA88}, can be used. 
The OWA criterion allows the decision maker to express 
information about importance of each objective (scenario)
by assigning some weights to them and
utilize   this information while computing a solution.
Some special cases of  the OWA criterion are the maximum, average, median, or Hurwicz criteria, which are the traditional criteria used in multiobjective optimization or in
decision making under uncertainty~(see, e.g.,~\cite{ER05,LR57}). 
A detailed description of various properties and applications of OWA can be found in~\cite{YKB11}.

Unfortunately, minimizing OWA is NP-hard for most basic combinatorial optimization problems, even if the number of objectives equals two. This negative result also holds for the  class of robust min-max problems~\cite{KY97, KZ16b}, being a special case of OWA minimization. The problem complexity increases with the number of objectives. It turns out that the problem of minimizing OWA if the number of objectives is a part of the input
 is not at all approximable for basic network problems~\cite{KZ15}. 
 The case of nonincreasing weights in the OWA  criterion,  that models preferences with respect to a risk, 
  is more computationally tractable.  However,  it is still NP-hard in general, for example,  when the OWA  becomes
  the maximum criterion (decision maker is extremely risk-averse). 
  Fortunately, for this case
  some approximation algorithms are known~\cite{KZ15}.

The most popular method of solving a multiobjective optimization problem with the OWA criterion is to formulate it as a mixed integer program (MIP) and apply some mathematical programming techniques to find an optimal solution.  Several such formulations were proposed in~\cite{OS03, OO12}. They were further investigated and refined in~\cite{FPP14, FPPS17, CG15, GS12}. The MIP approach is efficient when  problem sizes (in particular the number of objectives) are not  very large. For basic combinatorial optimization problems, such as the shortest path, minimum matching, or minimum spanning tree, computational test were performed for instances having up to~10 objectives~\cite{FPP14, FPPS17,GS12}. 
For larger number of objectives, 
the MIP approach can be inefficient and thus there is need  for methods
 which provide good approximate solutions. 
Problems with larger number of objectives can arise in optimization under uncertainty, when, for instance,
a scenario analysis (simulation) is performed (see, e.g.,~\cite{KW94,RW91,W89}).
A sample of parameter 
values (cost values) from any distribution is generated and models a correlation among these parameters by their 
joint realizations (scenarios). Thus the uncertainty about the parameters is modeled  by
a set of scenarios and, obviously,
the bigger the set is, the better is the estimation of the uncertainty.

In this paper we  propose a method of reducing the number of objectives before solving the problem. The reduction consists in partitioning the objectives into groups of  size $\ell$ and  replacing each group with one objective, by appropriately aggregating the costs. For example, when $\ell=2$, we can reduce the number of objectives by half, which can significantly decrease the time of solving the problem by using, for example, a MIP formulation solved by an off-the-shelf software. We will show that after solving such a smaller-sized problem, we get an approximate solution with some guaranteed worst-case approximation ratio. We can thus significantly reduce the problem size with only a small decrease of the quality of a  solution obtained. 
In the extreme case, we can aggregate all objectives into one. We prove that an optimal solution to such a problem has the same approximation guarantee as the solutions obtained by using the approximation algorithm proposed in~\cite{KZ15}. We show the results of experiments, which suggest that practical performance of the aggregation 
proposed may be better than the theoretical worst one.
The idea of aggregation of the objectives has been recently applied to the class of robust min-max (regret) problems in~\cite{CG18}. Since minimizing OWA generalizes the min-max approach, some results obtained in~\cite{CG18} are extended in this paper.
It is worth pointing out that the aforementioned 
aggregations assume any order of objective functions and 
they do not exploit a similarity of the objectives. 
Accordingly,
 in the paper we also propose a heuristic method,
based on a $K$-means  clustering, which takes this fact into account,
namely, ``similar'' objective functions   are aggregated together. 
Unfortunately, such an aggregation does not possess a guaranteed worst-case approximation ratio.
Thus we evaluate it experimentally.

This paper is organized as follows. In Section~\ref{sec1} we present the problem formulation and recall its known computational properties. In Section~\ref{sec2} we show that the aggregation approaches allow us to compute approximate solutions for a wide class of problems. In Section~\ref{sec3} we provide new approximation results for  the Hurwicz criterion. Finally, in Section~\ref{sec4}, we present the results of  computational tests, which give some evidence that our aggregation approaches, together with an off-the-shelf MIP solver, gives good
approximate solutions in reasonable time.

\section{Problem formulation}
\label{sec1}

In this paper we are concerned with the following multiobjective combinatorial 
optimization problem~$\mathcal{P}$ with $K$ linear objective functions:
\[
\begin{array}{lll}
\mathcal{P}: &\min &  (\pmb{c}^T_1\pmb{x},\dots,\pmb{c}^T_K\pmb{x})\\
&  & \pmb{x} \in \X \subseteq\{0,1\}^n,
 \end{array}
 \]
where $\X$ is a set of feasible solutions, typically described in the form of linear constraints,  and 
$\pmb{c}^T_k=(c_{1k},\dots, c_{nk})^T$ is a vector of nonnegative costs under the $k$th objective, 
$k\in [K]$ ($[K]$ denotes the set $\{1,\ldots,K\}$). 
The meaning of the $\min$ in  problem~$\mathcal{P}$
depends on a comparison among  objective value vectors
$(\pmb{c}^T_1\pmb{x},\dots,\pmb{c}^T_K\pmb{x})$ for feasible solutions 
 $\pmb{x} \in \X$. One of the most popular approaches  to 
 solving~$\mathcal{P}$  is reducing it to
  a problem with a single objective function, by using an aggregation function, and
 exploiting the natural ordering ``$\leq$'' in~$\R$ (see, e.g.,~\cite{ER05}).

In this paper 
we  aggregate the objectives  in~$\mathcal{P}$ by the OWA operator, proposed in~\cite{YA88}, which is defined as follows. 
Let $\pmb{a}=(a_1,\dots,a_K)$ be 
a vector of  reals. We  define a vector of weights $\pmb{w}=(w_1,\dots,w_K)$ to be $w_1+\dots+w_K=1$ and $w_k\in [0,1]$ for each $k\in [K]$. Let $\sigma$ be a permutation of $[K]$ such that $a_{\sigma(1)}\geq a_{\sigma(2)}\geq \dots \geq a_{\sigma(K)}$. Then
$$\mathrm{OWA}_{\pmb{w}}(\pmb{a})=\sum_{k\in [K]} w_k a_{\sigma(k)}.$$
We now discuss several special cases of OWA
that are well-known criteria used in multiobjective optimization or in
decision making under uncertainty.

Consider first OWA with nonincreasing weights, i.e. when $w_1\geq w_2\geq \dots \geq w_K$.
If $w_1=1$ and $w_k=0$ for $k>1$, then OWA becomes the maximum. If the weights are uniform, i.e. $w_k=\frac{1}{K}$ for each $k\in [K]$, then OWA is the average (or the Laplace criterion). The maximum and average are extreme cases of OWA with nonincreasing weights. An intermediate case is the $p$-centra criterion, $p\in[K]$, obtained by the following weight setting
$w_k=\frac{1}{p}$ for $k\in[p]$ and  $w_k=0$ for $k>p$. 

We now turn to OWA with arbitrary weights.
If $w_{\lceil K/2 \rceil}=1$ and $w_k=0$ for $k\neq \lceil K/2 \rceil$, then OWA becomes the median.
Another important case is when $w_1=\lambda$, $w_K=1-\lambda$ for some $\lambda\in [0,1]$ and $w_k=0$ for the remaining weights, which corresponds to well known Hurwicz  pessimism-optimism criterion, being a convex combination of the maximum and minimum  possible objective values. Finally, if $w_K=1$ and $w_k=0$ for $k<K$, then OWA becomes the minimum.

After applying the OWA criterion with specified weights~$\pmb{w}$ to aggregate the vector of objective values 
$\pmb{F}(\pmb{x})=(\pmb{c}^T_1\pmb{x},\dots,\pmb{c}^T_K\pmb{x})$ for a given feasible solution~$\pmb{x}\in \mathcal{X}$, we get  the following optimization problem, considered in this paper:
$$\textsc{OWA}~\mathcal{P}: \min_{\pmb{x}\in \mathcal{X}} \mathrm{OWA}_{\pmb{w}}(\pmb{F}(\pmb{x}))=\min_{\pmb{x}\in \mathcal{X}} \mathrm{OWA}_{\pmb{w}}(\pmb{c}^T_1\pmb{x},\dots,\pmb{c}^T_K\pmb{x}).$$
The $\textsc{OWA}~\mathcal{P}$ problem arises naturally  in the robust optimization setting. Each cost 
vector $\pmb{c}_k$,
 $k\in [K]$, can be then interpreted as a \emph{scenario}, i.e. a realization of the uncertain costs
 (a state of the world) which can occur. In this case, nonincreasing weights model a  risk aversion of decision makers, i.e. the less uniform are the weights the more risk averse decision maker is. In the extreme case 
 ($w_1=1$, $w_k=0$ for $k>1$) $\textsc{OWA}~\mathcal{P}$ becomes the robust min-max version of problem $\mathcal{P}$, which is widely discussed in the existing literature.

Let us now briefly describe the complexity of $\textsc{OWA}~\mathcal{P}$. It is clear that this problem is NP-hard when the corresponding problem with one objective is already NP-hard.  However, it is well known that $\textsc{OWA}~\mathcal{P}$ is NP-hard for most basic polynomially solvable problems $\mathcal{P}$, even if $K=2$. This is a direct consequence of the results obtained for robust min-max problems (see~\cite{KY97, ABV09, KZ16b} for surveys). 
For arbitrary weight vectors $\pmb{w}$, $\textsc{OWA}~\mathcal{P}$ cannot be approximated for some network problems~\cite{KZ15} (for example when $\mathcal{P}$ is the shortest path problem). This negative result 
holds  when, for example,  OWA is the median~\cite{KZ15}. However, the problem is more tractable when the weights are nonincreasing. In this case the following approximation algorithm has been proposed in~\cite{KZ15}. Define $\hat{c}_i=\mathrm{OWA}_{\pmb{w}}(c_{i1},\dots, c_{iK})$ for $i\in [n]$. We solve problem $\mathcal{P}$ with one  aggregated objective $\hat{\pmb{c}}^T\pmb{x}$, where
 $\hat{\pmb{c}}^T=(\hat{c}_1,\dots,\hat{c}_n)^T$. It has been shown in~\cite{KZ15}, that under nonincreasing weights this algorithm has an approximation ratio of~$w_1K$, which is the best currently known general result. For some particular cases (for example when $\mathcal{P}$ is the shortest path) $\textsc{OWA}~\mathcal{P}$ admits an FPTAS if the number of objectives~$K$ is constant~\cite{KZ15}. However, from the practical point of view, the FPTAS is inefficient as its running time is exponential in~$K$.

In Section~\ref{sec2} we discuss the $\textsc{OWA}~\mathcal{P}$ problem with nonincreasing weights. We apply the idea of reducing the number of objectives before solving the problem. This approach has been originally proposed for the robust min-max problems in~\cite{CG18}. We show how it can be extended to $\textsc{OWA}~\mathcal{P}$.

\section{The problem with nonincreasing weights}
\label{sec2}

In this section 
we make the assumption that the weights $\pmb{w}$ in $\textsc{OWA}~\mathcal{P}$
are such that $w_1\geq w_2\geq\cdots \geq w_K$. 
Notice that this case
contains both the maximum and the average criteria as special (boundary) cases
and allows decision makers to take their attitude towards a risk into account. We present  an aggregation of the objective values in  $\textsc{OWA}~\mathcal{P}$
that allows us to reduce the number of the objective functions in the problem under consideration
and compute a solution with a guaranteed approximation ratio.

\subsection{Auxiliary results}
\label{sec2-1}

We start by proving some properties of the aggregating  approach.
Let us recall Chebyshev's sum inequality (see, eg.,~\cite[p. 36]{M70}), namely:
\begin{lemma}
\label{lem0}
Let $(a_1,\dots,a_m)$ and $(w_1,\dots,w_m)$ be two real vectors such that 
$a_1\geq a_2\geq \dots \geq  a_m$ and $w_1\geq w_2\geq \dots \geq w_m$. 
Then the following inequality 
$$(w_1+w_2+\dots+w_m)(a_1+a_2+\dots+a_m)\leq m (w_1a_1+w_2a_2+\dots+w_ma_m)$$
holds.
\end{lemma}
\begin{lemma}
\label{lem01}
Let $\pmb{a}=(a_1,\dots,a_K)$  be a nonnegative real vector  and $\pmb{w}$ be 
a nonincreasing weight vector. Suppose that $a_{\sigma(1)}\geq a_{\sigma(2)}\geq\cdots\geq a_{\sigma(K)}$,
where $\sigma$ is a permutation of~$K$.
Let us construct $\pmb{a}'$ by setting  $a'_{\sigma(i)}:=a_{\sigma(i)}+a_{\sigma(j)}$ and $a'_{\sigma(j)}:=0$ 
for some $\sigma(i)<\sigma(j)$ in $\pmb{a}$. 
Then $\mathrm{OWA}_{\pmb{w}}(\pmb{a}')\geq \mathrm{OWA}_{\pmb{w}}(\pmb{a})$.
\end{lemma}
\begin{proof}
Since $w_i\geq w_j$, we get
\begin{align*}
\mathrm{OWA}_{\pmb{w}}(\pmb{a}) &= \sum_{k\in [K]\setminus\{i,j\}} w_k a_{\sigma(k)}+(w_ia_{\sigma(i)}+w_ja_{\sigma(j)}) \\
&\leq \sum_{k\in [m]\setminus\{i,j\}} w_k a_{\sigma(k)}+(w_i(a_{\sigma(i)}+a_{\sigma(j)}) + w_j \cdot 0)\leq  \mathrm{OWA}_{\pmb{w}}(\pmb{a}').
\end{align*}
\end{proof}

Consider a  nonzero vector of nonnegative reals  $\pmb{a}=(a_1,a_2,\dots,a_K)$ and
a vector of 
nonincreasing weights $\pmb{w}$.
Assume that $K$ is a multiple of $\ell$ and   the component values in $\pmb{a}$ are in an
 arbitrary order. We now aggregate the  values in $\pmb{a}$ and $\pmb{w}$
 and form corresponding vectors  $\overline{\pmb{a}}$ and $\overline{\pmb{w}}$
 whose sizes are reduced to $K/\ell$. Namely, let
 $\overline{\pmb{a}}=(\overline{a}_1,\overline{a}_2,\dots,\overline{a}_{K/\ell})$ 
and $\overline{\pmb{w}}=(\overline{w}_1,\dots,\overline{w}_{K/\ell})$, where
 $\overline{a}_k=\frac{1}{\ell}(a_{(k-1) \ell+1}+\dots+a_{k \ell})$ and
$\overline{w}_k=w_{(k-1) \ell+1}+\dots+w_{k \ell}$ for $k\in [K/\ell]$, i.e. $\overline{a}_k$ is 
formed by averaging $\ell$ subsequent values 
and $\overline{w}_k$ is the sum of $\ell$ subsequent weights. 
The next lemma characterizes the value of $\mathrm{OWA}$ for   aggregated vector~$\overline{\pmb{a}}$.
\begin{lemma}
\label{lem1}
For nonincreasing weights~$\pmb{w}$ and 
 any  nonzero vector of nonnegative reals  $\pmb{a}=(a_1,a_2,\dots,a_K)$,
the following inequalities
\begin{equation}	
	\mathrm{OWA}_{\pmb{\overline{w}}}(\pmb{\overline{a}})\overset{(a)}{\leq} \mathrm{OWA}_{\pmb{w}}(\pmb{a})\overset{(b)}{\leq}  \ell\rho \cdot \mathrm{OWA}_{\pmb{\overline{w}}}(\pmb{\overline{a}})
	\label{lem1ab}
\end{equation}	
hold,
where $\rho=\max_{k\in [K/ \ell]}\frac{\sum_{i=1}^k w_i}{\sum_{i=1}^{k} {\overline{w}_i}}$.
\end{lemma}
\begin{proof}
Let us renumber the elements of $\pmb{a}$ so that $\overline{a}_1\geq \overline{a}_2\geq \cdots \geq \overline{a}_{K/\ell}$. 
Futhermore, within each component $\overline{a}_k$, the elements are indexed so that they form a nonincreasig sequence. For example,  let $\pmb{a}=(5,1,3,6,0,6,2,0,1)$ and $\ell=3$. After aggregation we get $\frac{1}{3}(5+1+3) = 3$, $\frac{1}{3}(6+0+6) = 4$, and $\frac{1}{3}(2+0+1) = 1$. 
Hence, after renumbering the elements, we get $\pmb{a}=(6,6,0, 5,3,1,2,1,0)$. Let $J=\{1,\ell+1,2\ell+1,\dots,K-\ell+1\}$.  
Thus we have
$$\mathrm{OWA}_{\pmb{\overline{w}}}(\pmb{\overline{a}})=
\frac{1}{\ell}\sum_{k\in J}(w_k+\dots+w_{k+\ell-1})(a_k+\dots+a_{k+\ell-1}).
$$
Let us first prove inequality~(\ref{lem1ab}a). From Lemma~\ref{lem0}, we get
$$\mathrm{OWA}_{\pmb{\overline{w}}}(\pmb{\overline{a}})
 \leq \sum_{k\in J}(w_ka_k+\dots+w_{k+\ell-1}a_{k+\ell-1})=\sum_{k\in [K]} w_k a_k\leq \mathrm{OWA}_{\pmb{w}}(\pmb{a}).$$
We now prove  inequality~(\ref{lem1ab}b).
Consider component $(a_{(k-1)\ell+1}+\dots+a_{k \ell})$, $k\in [K/\ell]$.  
Define $a'_{(k-1)\ell+1}=a_{(k-1)\ell+1}+\dots+a_{k \ell}$ and $a'_{(k-1)\ell+2}=\dots=a'_{k \ell}=0$. 
For $\pmb{a}=(6,6,0, 5,3,1,2,1,0)$, we get $\pmb{a}'=(12,0,0,9,0,0,3,0,0)$. It is easy to see that $\overline{\pmb{a}} = \overline{\pmb{a}}'$ and therefore $\mathrm{OWA}_{\pmb{\overline{w}}}(\pmb{\overline{a}})=\mathrm{OWA}_{\pmb{\overline{w}}}(\pmb{\overline{a}}')$. 
Using iteratively Lemma~\ref{lem01}, we obtain
$\mathrm{OWA}_{\pmb{w}}(\pmb{a}')\geq \mathrm{OWA}_{\pmb{w}}(\pmb{a})$.
This yields
$$\mathrm{OWA}_{\pmb{w}}(\pmb{a})\leq \mathrm{OWA}_{\pmb{w}}(\pmb{a}')=w_1a_1'+w_2a_{\ell+1}'+\dots+w_{K/\ell}a_{K-\ell+1}'$$
since $a_1'\geq a_{\ell+1}'\geq\dots\geq a_{K-\ell+1}'$ and $a'_k=0$ for the remaining elements. 
Let us rewrite
 $$\ell\mathrm{OWA}_{\pmb{\overline{w}}}(\pmb{\overline{a}}')=\overline{w}_1a_1'+\overline{w}_2a_{\ell+1}'+\dots+\overline{w}_{K/\ell}a'_{K-\ell+1}.$$
 We now estimate  from above the following ratio: 
 $$\frac{\mathrm{OWA}_{\pmb{w}}(\pmb{a}')}{\ell\mathrm{OWA}_{\pmb{\overline{w}}}(\pmb{\overline{a}'})}=\frac{w_1a_1'+w_2a_{\ell+1}'+\dots+w_{K/\ell}a_{K-\ell+1}'}{\overline{w}_1a_1'+\overline{w}_2a_{\ell+1}'+\dots+\overline{w}_{K/\ell}a'_{K-\ell+1}}.$$
We use the following fractional programming problem:
\begin{equation}
\label{fr0}
 	\begin{array}{lll}
	\max z= & \frac{w_1x_1+w_2x_2+\dots+w_{K/\ell}x_{K/\ell}}{\overline{w}_1x_1+\overline{w}_2x_2+\dots+\overline{w}_{K/\ell}x_{K/\ell}} \\
	& x_{k+1}-x_k\leq 0 & k\in [K/\ell-1]\\
	& x_k\geq 0 & k\in [K/\ell]
	\end{array}
 \end{equation}
where $x_1$ corresponds to $a_1'$, $x_2$ to $a_{\ell+1}'$ etc. The constraints ensure that $x_1\geq x_2\geq \dots \geq x_{K/\ell}$ form a nonincreasing sequence of nonnegative numbers.
Under the assumption that $(x_1,\dots,x_{K/\ell})$ is a nonzero vector,  problem (\ref{fr0}) is equivalent to the following linear programming program (with dual variables in brackets):
 $$
 	\begin{array}{llll}
	\max z= & w_1x_1+w_2x_2+\dots+w_{K/\ell}x_{K/\ell} \\
	& x_{k+1}-x_k\leq 0 & k\in [K/\ell-1] & [\alpha_k]\\
	&\overline{w}_1x_1+\overline{w}_2x_2+\dots+\overline{w}_{K/\ell} x_{K/\ell}=1 & & [\gamma]\\
	& x_k\geq 0 & k\in [K/\ell]
	\end{array}
 $$
 The dual is
  $$
 	\begin{array}{llll}
	\min z'= & \gamma \\
	& \gamma\overline{w}_1-\alpha_1\geq w_1 & [x_1]\\
	& \gamma\overline{w}_2+\alpha_1-\alpha_2\geq w_2 & [x_2]\\
	& \gamma\overline{w}_3+ \alpha_2 -\alpha_3 \geq w_3 & [x_3] \\
	& \vdots \\
	& \gamma \overline{w}_{K/\ell}+ \alpha_{K/l-1} \geq w_{K/\ell} & [x_{K/\ell}] \\
	&\alpha_i\geq 0 & i\in [K/\ell-1]
	\end{array}
 $$
Assume that $(x_1^*,x_2^*,\dots,x_{K/\ell}^*)$ is an optimal primal solution and $(\gamma^*, \alpha^*_1,\dots,\alpha^*_{K/\ell-1})$ is an optimal dual solution. If all the primal variables are positive, then according 
to the complementary slackness condition (see, e.g.,~\cite{PS98}), all 
the dual constraints must be tight. Adding them, we get 
$z'=\gamma^*=(w_1+\dots+w_{K/\ell}) / (\overline{w}_1+\dots+\overline{w}_{K/\ell})\leq \rho$.
 Let $x_1^*,x_2^*,\dots,x^*_t>0$ and $x^*_{t+1}=\dots =x^*_{K/\ell}=0$ for some  $t\geq 1$ and $t<[K/\ell]$. 
 Again by the complementary slackness condition, the first $t$ constraints in the dual must be tight. Furthermore, because $x^*_{t+1}-x^*_{t}<0$, the dual variable $\alpha^*_t=0$. Adding the first $t$ dual constraints yields
$$\gamma^*(\overline{w}_1+\overline{w}_2+\dots+\overline{w}_{t})=w_1+w_2+\dots + w_t,$$
Hence $z'=\gamma^*=\frac{\sum_{k=1}^t w_k}{\sum_{k=1}^{t}\overline{w}_k}\leq \rho$, since $t\in [K/\ell]$. Finally
$$
\frac{\mathrm{OWA}_{\pmb{w}}(\pmb{a})}{\ell\mathrm{OWA}_{\pmb{\overline{w}}}(\pmb{\overline{a}'})}\leq
\frac{\mathrm{OWA}_{\pmb{w}}(\pmb{a}')}{\ell\mathrm{OWA}_{\pmb{\overline{w}}}(\pmb{\overline{a}'})} \leq \rho
$$
and  inequality~(\ref{lem1ab}b) holds. 
\end{proof}

\subsection{Aggregation algorithm}
\label{sec2-2}

We are now ready to apply the aforementioned aggregation to
the problem $\textsc{OWA}~\mathcal{P}$ with nonincreasing weights $w_1\geq w_2\geq \dots \geq w_K$. 
We can assume that $K$ is a multiple of $\ell\geq 1$;  otherwise we add a necessary number of dummy objectives with 0 
costs for all variables and the weights in OWA criterion equal to~0. 
Let $\pmb{C}$ be the $n\times K$ matrix with the columns $\pmb{c}_1,\dots,\pmb{c}_K$. Let us construct an
$n\times K/\ell$ matrix $\overline{\pmb{C}}$ by aggregating the columns of $\pmb{C}$. 
Namely, $\overline{\pmb{C}}$ has columns $\overline{\pmb{c}}_1,\dots,\overline{\pmb{c}}_{K/\ell}$, where $\overline{\pmb{c}}_k=\frac{1}{\ell}(\pmb{c}_{(k-1) \ell+1}+\dots+{\pmb{c}}_{k \ell})$, $k\in [K/\ell]$.
An example for $\ell=2$ is shown in Table~\ref{tab1}. After the aggregation we get  
$\overline{\pmb{F}}(\pmb{x})=(\overline{\pmb{c}}^T_1\pmb{x}, \dots,\overline{\pmb{c}}^T_{K/l}\pmb{x})$, 
where $\overline{\pmb{c}}^T_k\pmb{x}=\frac{1}{\ell}(\pmb{c}^T_{(k-1) \ell+1}\pmb{x}+\dots+{\pmb{c}}^T_{k \ell}\pmb{x})$, 
$k\in [K/\ell]$.
Assume that $\overline{\pmb{x}}$ minimizes $\mathrm{OWA}_{\overline{\pmb{w}}}(\overline{\pmb{F}}(\pmb{x}))$.
The following lemma characterizes the computed solution~$\overline{\pmb{x}}$.
\begin{lemma}
\label{thm1}
Given any $\pmb{x}\in \X$ and nonincreasing weights~$\pmb{w}$. Then
\begin{equation}
    \mathrm{OWA}_{\pmb{w}}(\pmb{F}(\overline{\pmb{x}}))\leq \ell\rho \cdot \mathrm{OWA}_{\pmb{w}}(\pmb{F}(\pmb{x})).
\end{equation}
\end{lemma}
\begin{proof}
By Lemma~\ref{lem1} we get 
$$\mathrm{OWA}_{\pmb{w}}(\pmb{F}(\pmb{x}))\geq \mathrm{OWA}_{\overline{\pmb{w}}}(\overline{\pmb{F}}(\pmb{x}))\geq \mathrm{OWA}_{\overline{\pmb{w}}}(\overline{\pmb{F}}(\overline{\pmb{x}}))\geq \frac{1}{\ell\rho} \cdot \mathrm{OWA}_{\pmb{w}}(\pmb{F}(\overline{\pmb{x}}))$$
and the lemma follows.
\end{proof}

We are thus led to the following
$\ell$-\textsc{Aggregation Algorithm}: given an instance of $\textsc{OWA}~\mathcal{P}$ with nonincreasing weights $w_1\geq w_2\geq \dots \geq w_K$ and $K$ objectives $\pmb{c}^T_1\pmb{x},\dots,\pmb{c}^T_K\pmb{x}$, solve the corresponding instance of $\textsc{OWA}~\mathcal{P}$ with weights $\overline{w}_1\geq \dots \geq \overline{w}_{K/\ell}$ and objectives 
$\overline{\pmb{c}}^T_1\pmb{x},\dots,\overline{\pmb{c}}^T_{K/\ell}\pmb{x}$. 
From Lemma~\ref{thm1} we immediately get the following result:
\begin{theorem}
\label{thm2}
	The $\ell$-\textsc{Aggregation Algorithm} has an approximation ratio of $\rho \ell$, where $\rho=\max_{k\in [K/\ell]}\frac{\sum_{i=1}^k w_i}{\sum_{i=1}^{k} {\overline{w}_i}}$.
\end{theorem}
Let us now analyze the quality of a  solution returned by the $\ell$-\textsc{Aggregation Algorithm}.
 It is easy to check that  $\rho\in [\frac{1}{\ell},1]$. 
 The bound $\rho\geq \frac{1}{\ell}$ follows from the fact that $\frac{w_1}{\overline{w}_1}\geq \frac{1}{\ell}$ (which can be rewritten as $\ell w_1\geq (w_1+w_2+\dots+w_\ell)$).
Furthermore, $\rho=\frac{1}{\ell}$ if and only if the weights are uniform, i.e. $w_k=\frac{1}{K}$ for each $k\in [K]$. This corresponds to the case when OWA is the average. Of course, the aggregation preserves then the optimality of the solution.
On the other hand, $\rho=1$  if and only if $w_1+w_2+\dots+w_{K/\ell}=1$, i.e. when 1 is allocated to the largest 
$K/\ell$ 
(OWA is the $K/\ell$-centra criterion).  This is true, for example, when $w_1=1$, i.e. when OWA is the maximum. In this case, the algorithm returns an $\ell$-approximate solution, which is a generalization of the results from \cite{CG18}.
Observe also, that when $\ell=K$, then $\rho=w_1$ and the algorithm has an  approximation ratio of~$w_1K$, which is the same as the one obtained in~\cite{KZ15}. The aggregation for $\ell=K$ results in one objective, in which each cost is just the average cost over all objectives. Such an aggregation  is different than the one proposed in~\cite{KZ15}. However, both of them result in the same approximation guarantee $w_1K$.

We now give a brief numerical discussion of the quality of our bound~$\rho \ell$
(see Theorem~\ref{thm2}). Consider a problem with $K=200$. The weight vectors~$\pmb{w}$ are obtained using a function that takes a parameter $\alpha\in (0,1)$ (for more details, see~\eqref{genfun} in Section~\ref{sec4}). The larger this parameter is, the less distorted is the weight distribution, i.e. it becomes closer to a uniform weight distribution. 
Table~\ref{tab-theory} presents  values of  the bound for different values of $\ell$ and $\alpha$.
\begin{table}[h]
\caption{Values of the bound from Theorem~\ref{thm2} for $K=200$.}\label{tab-theory}
\centering\begin{tabular}{r|rrrrrrr}
 & \multicolumn{7}{c}{$\ell$} \\
 & 2 & 5 & 10 & 20 & 50 & 100 & 200\\
\hline
\multirow{ 3}{*}{$\alpha$} $10^{-2}$ & $1.52$ & $2.05$ & $2.29$ & $2.42$ &  $2.50$ & $2.53$ & $2.54$ \\
$10^{-3}$ & $1.94$ & $3.75$ & $4.99$ & $5.85$ & $6.46$ & $6.68$ & $6.80$ \\
$10^{-6}$ & $2.00$ & $4.68$ & $7.49$ & $9.98$ & $12.07$ & $12.90$ & $13.35$
\end{tabular}
\end{table}

We see at once that  the approximation guarantees  are better for less distorted weights. 
When $\alpha=10^{-6}$, the approximation algorithm proposed in~\cite{KZ15} 
has a worst-case ratio of~$13.35$. We can reduce it to $12.90$,  by choosing $\ell=100$. 
As the result, we get a problem with only two objectives, which can be solved to optimality in 
reasonable time by MIP solvers.
 We can also use smaller values of $\ell$. By setting~$\ell=2$, we  can reduce the number of objectives 
 by $50\%$ and after solving the obtained instance, we get a~2.00 approximate solution. If $\ell=5$, then we reduce the number of objectives  by $80\%$ and after solving the resulting instance we get a $4.68$ approximate solution.  We thus can see that the aggregation allows us to establish a trade-off between the running time of an exact algorithm and the quality the obtained solutions. Notice that the worst-case ratio is only theoretical and the $\ell$-\textsc{Aggregation Algorithm}  may behave much better in practice (we will test it in more detail in Section~\ref{sec4}).

We continue in this fashion and give
 a sample worst instance for the $\ell$-\textsc{Aggregation Algorithm}, where $\ell=2$.
Consider a problem with~4 variables and 8~objectives, shown in Table~\ref{tab1}. 
\begin{table}[ht]
\caption{(a) A sample problem with $n=4$,  $K=8$ and $\pmb{w}=(0.2,0.2,0.1,0.1,0.1,0.1,0.1,0.1)$. 
(b) The problem after the aggregation with $\ell=2$ with $\overline{\pmb{w}}=(0.4, 0.2, 0.2, 0.2)$.} \label{tab1}
\centering
\begin{tabular}{ccc}
\begin{tabular}{l l|ll | ll | ll | lll}
(a)	& & $\pmb{c}_1$ & $\pmb{c}_2$ & $\pmb{c}_3$  & $\pmb{c}_4$  & $\pmb{c}_5$  & $\pmb{c}_6$  & $\pmb{c}_7$  & $\pmb{c}_8$  \\
 \cline{2-10}
&$x_1$ & 1 & 0 & 1 & 0 & 0 & 0 & 0 & 0 \\
&$x_2$ & 0 & 1 & 0 & 1 & 0 & 0 & 0 & 0 \\ 
&$x_3$ & 1 & 0 & 1 & 0 & 0 & 0 & 0 & 0 \\ 
&$x_4$ & 0 & 1 & 0 & 1 & 0 & 0 & 0 & 0 \\
\end{tabular}
&
\begin{tabular}{ll | llll}
 (b)	& & $\overline{\pmb{c}}_1$ & $\overline{\pmb{c}}_2$ & $\overline{\pmb{c}}_3$  & $\overline{\pmb{c}}_4$   \\
\cline{2-6}	
 &$x_1$ &0.5 & 0.5 & 0 & 0 \\
&$x_2$ & 0.5 & 0.5 & 0 & 0 \\ 
&$x_3$ & 0.5 & 0.5 & 0 & 0  \\ 
&$x_4$ & 0.5 & 0.5 & 0 & 0 \\
\end{tabular}
\end{tabular}
\end{table}
Assume that $\mathcal{X}$ contains all 0-1 solutions satisfying the constraint $x_1+x_2+x_3+x_4=2$.
 After the aggregation, each feasible solution has the same cost under each objective (see Table~\ref{tab1}(b)).
Choose solution $x_1=1$, $x_2=1$, $x_3=0$, $x_4=0$ with 
$\pmb{F}(\pmb{x})=(1,1,1,1,0,0,0,0)$ and $\mathrm{OWA}_{\pmb{w}}(\pmb{F}(\pmb{x}))=w_1+w_2+w_3+w_4$. 
On the other hand, for solution $x'_1=1$, $x'_3=1$, $x'_2=0$, $x'_4=0$ we have $\pmb{F}(\pmb{x}')=(2,0,2,0,0,0,0,0)$ and $\mathrm{OWA}_{\pmb{w}}(\pmb{F}(\pmb{x}'))=2(w_1+w_2)$. Hence, $\mathrm{OWA}_{\pmb{w}}(\pmb{F}(\pmb{x}'))=2\rho \mathrm{OWA}_{\pmb{w}}(\pmb{F}(\pmb{x}))$, where $\rho=\frac{w_1+w_2}{w_1+w_2+w_3+w_4}$. For the sample weight vector shown in Table~\ref{tab1}, we get $\rho=\frac{2}{3}$ and the 2-\textsc{Aggregation Algorithm} may return a $\frac{4}{3}$-approximate solution.
 It is not difficult to extend this bad example for any $\ell$.

The $\ell$-\textsc{Aggregation Algorithm} allows us to establish another theoretical approximation result, which is analogous to the one obtained in~\cite{CG18}.
It has been shown in~\cite{KZ15} that when the number of objectives is constant, then $\textsc{OWA}~\mathcal{P}$ admits an FPTAS for some particular problems $\mathcal{P}$ (for example, when $\mathcal{P}$ is the shortest path problem). This means, among others, that if $K$ is constant, then $\textsc{OWA}~\mathcal{P}$ has a polynomial 2-approximation algorithm.
Without loss of generality we can assume that that $K=2^r$ for some $r>1$. 
 We will consider $\ell$-\textsc{Aggregation Algorithm} for $\ell=2,4,\dots,2^r$. If $\ell=2^{r-j}$ then we will are at \emph{$j$th level} of aggregation. Using Theorem~\ref{thm2}, we get that the $j$th level of aggregation gives us $2^{r-j}\rho=(\rho /2^{j})K$ approximate solution. Let us fix a constant $\epsilon\in (0,1)$ and choose $\overline{j}=\lceil \log(1/\epsilon)+1\rceil$. As a result we get a problem with a fixed number of objectives equal to $2^{\overline{j}}$. We can now apply a 2-approximation algorithm to this problem obtaining a $2\cdot (\rho /2^{\overline{j}})K\leq \epsilon K$-approximate solution (since $\rho\leq 1$). The following theorem summarizes the above reasoning:
\begin{theorem}
\label{thm3}
If $\textsc{OWA}~\mathcal{P}$ has a polynomial 2-approximation algorithm for a constant number of objectives, then it also has a polynomial $\epsilon K$-approximation algorithm for each constant $\epsilon>0$.
\end{theorem}
Theorem~\ref{thm3} is a theoretical result. It could be useful if we had more efficient 2-approximation algorithm for a fixed number of objectives. Currently,  2-approximation algorithms are based on the existence of FPTAS for fixed~$K$, which unfortunately are exponential in~$K$ (see~\cite{ABV10}).

\subsection{Heuristic aggregation}
\label{sec2-3}

The results from Section~\ref{sec2-2} apply to any order of objectives. In practice, it may be reasonable to order the objectives in a way that ''similar'' functions are aggregated together. This fact can be observed for the bad instance shown in Table~\ref{tab1}. If we exchange vectors $\pmb{c}_2$ and $\pmb{c}_3$ in Table~\ref{tab1}(a) and again aggregate with $\ell=2$, then the optimality of solution will be preserved. The reason is that exactly the same two objectives are then aggregated in every group.  In this section we propose a heuristic aggregation method, which does not fit into the theoretical framework presented in the previous section, but will be used for comparison in the experimental section.

Different approaches can be used to define similarity of objective functions. In the following, we will use the Euclidean norm $d(\pmb{c}_i,\pmb{c}_j) = \| \pmb{c}_i - \pmb{c}_j \|_2$ between the respective objective function coefficients. Given $K$ objectives, the aim is to aggregate them to a specified target value $\overline{K}<K$ of objectives. Namely, we wish to form $\overline{K}$ groups (clusters) of objectives $C_1,\dots, C_{\overline{K}}$ to minimize $\sum_{i\in [\overline{K}]} \sum_{\pmb{c}\in C_i} ||\pmb{c}-\pmb{\mu}_i||_{2}$, where $\pmb{\mu}_i=\frac{1}{|C_i|}\sum_{\pmb{c}\in C_i} \pmb{c}_i$ is the mean point in $C_i$.  While this problem is known to be NP-hard, strong heuristics for this purpose are readily available (see, e.g., \cite{jain2010data}). To this end, we apply a $\overline{K}$-means clustering algorithm to find $\overline{K}$ groups of similar objectives which are then aggregated, i.e. each $C_i$ is replaced with $\pmb{\mu}_i$, $i\in[\overline{K}]$.

The objective weights $\pmb{w}$ are then aggregated as uniformly as possible. Let $a,b\in\mathbb{N}_0$ be such that $K=a\overline{K}+b$ with $b<\overline{K}$. Then the first $b$ aggregated weights consist of $a+1$ original weights, while the remaining aggregated weights consist of $a$ original weights. More formally, for each $i\in[\overline{K}]$ we set
\[ \overline{w}_i = \begin{cases} \sum_{j=(i-1)(a+1)+1}^{i(a+1)} w_i & \text{ if } i\le b \\[2ex]
\sum_{j=b(a+1)+(i-b-1)a+1}^{b(a+1)+(i-b)a} w_i & \text{ otherwise } \end{cases} \]
For example, if $K=6$ and $\overline{K}=4$, then $\overline{\pmb{w}} = (w_1+w_2, w_3+w_4, w_5,w_6)$. Note that we can aggregate to any desired $\overline{K}$, without using dummy objectives. This is an advantage over the $\ell$-\textsc{Aggregation Algorithm} from Section~\ref{sec2-2}, where the size $\ell$ of the clustering is given. However, the clusters $C_i$ may have different cardinalities. Hence, the results from Section~\ref{sec2-2} cannot be applied to analyze this method.
 In fact, this aggregation may have worse theoretical approximation guarantee, as the example shown in Table~\ref{tab2} demonstrates.

\begin{table}[ht]
\caption{(a) A sample problem with $n=2$, $K$ even, and $\pmb{w}=(1/K,\ldots,1/K)$.\\
(b) The problem after the heuristic aggregation using $\overline{K}$-means with $\overline{K}=2$, and $\overline{\pmb{w}}=(0.5, 0.5)$.} \label{tab2}
\centering
\begin{tabular}{cccc}
\begin{tabular}{ll | lllll}
(a)	& & $\pmb{c}_1$ & $\pmb{c}_2$ & $\pmb{c}_3$ & $\ldots$  & $\pmb{c}_K$ \\
 \cline{2-7}
&$x_1$ & 1 & 0 & 0 & $\ldots$ & 0\\
&$x_2$ & 0 & 1 & 1 & $\ldots$ & 1
\end{tabular}
&
\begin{tabular}{ll | ll}
 (b)	& & $\overline{\pmb{c}}_1$ & $\overline{\pmb{c}}_2$ \\
\cline{2-4}	
 &$x_1$ & 1 & 0 \\
&$x_2$ & 0 & 1 \\
\end{tabular}
\end{tabular}
\end{table}

In the sample problem, we have two variables with one constraint $x_1+x_2=1$ and even $K$. There is one objective, where $x_1$ has a cost of~1, and $x_2$ has a cost of~0. In the remaining $K-1$ objectives $x_1$ has cost~0, and $x_2$ has cost~1. The weight vector $\pmb{w}$ is assumed to be uniform. It is easy to verify that the optimal solution to this problem is $\pmb{x}=(1,0)$ with $\mathrm{OWA}_{\pmb{w}}(\pmb{F}(\pmb{x}))=\frac{1}{K}$. The $\ell$-\textsc{Aggregation Algorithm} from Section~\ref{sec2-2} with $\ell=K/2$, gives us again the solution $\pmb{x}=(1,0)$, which follows from the fact that the weights in $\pmb{w}$ are uniform.
If the $\overline{K}$-means approach is used with $\overline{K}=2$, the vectors $\pmb{c}_2,\ldots,\pmb{c}_K$ are found to belong to the same cluster, so we end up with the problem shown in Table~\ref{tab2}(b). Now solutions $\pmb{x}=(1,0)$ and $\pmb{x}'=(0,1)$ have the same objective value for any aggregated weight vector $\overline{\pmb{w}}$ (the $\overline{K}$-means algorithm gives us the vector $\overline{\pmb{w}}=(0.5,0.5)$). If we choose $\pmb{x}'$, then $\mathrm{OWA}_{\pmb{w}}(\pmb{F}(\pmb{x}'))=\frac{K-1}{K}$, which is $K-1$ times worse than the optimum.

The fact that the $\overline{K}$-means approach has bad theoretical worst case ratio follows from the fact that the formed clusters can have different cardinalities. Nevertheless, there may be still a practical advantage of this method and we will explore it in the experimental section. Note that, if the theoretical guarantee is still required, one can also modify the $\ell$-\textsc{Aggregation Algorithm} by resorting the objectives into clusters of size $\ell$.

\section{The Hurwicz criterion}
\label{sec3}

In this section we show that the idea of aggregation can also be applied to 
 $\textsc{OWA}~\mathcal{P}$
when OWA is the 
pessimism-optimism
Hurwicz criterion, i.e. when $w_1=\lambda$, $w_K=1-\lambda$, and $w_k=0$ if $k\neq 1$ and $k\neq K$ for a fixed $\lambda\in [0,1]$.  This criterion is used  in decision under complete uncertainty and enables to take account the decision maker's 
attitudes that are neither extremely pessimistic nor extremely optimistic.
Notice the the weights are then not monotone, so the results obtained in Section~\ref{sec2} cannot be applied directly. Thus,
the Hurwicz  criterion  (see, e.g.,~\cite{LR57}) is a convex combination of the maximal and the minimal cost 
of~$\pmb{x}\in  \X$ in the set of $K$ objectives (scenarios) $\{\pmb{c}^T_1\pmb{x},\dots,\pmb{c}^T_K\pmb{x}\}$ and
OWA has the following form:
\begin{equation}
 \mathrm{OWA}_{\pmb{w}}(\pmb{F}(\pmb{x}))
  = \lambda\max_{k\in[K]} \pmb{c}^T_k\pmb{x} + (1-\lambda)\min_{i \in[K]} \pmb{c}^T_i\pmb{x}.
  \label{owahc}
\end{equation} 
The problem $\textsc{OWA}~\mathcal{P}$ with  the Hurwicz criterion~(\ref{owahc}) as a special case of~OWA
will be denoted by $\textsc{Hurwicz}~\mathcal{P}$.
It has been shown in~\cite{KZ15} that there exists a
$K/\lambda$ approximation
 algorithm for  $\textsc{Hurwicz}~\mathcal{P}$ when 
  $\lambda\in(0,\frac{1}{2})$, and  $\lambda K + (1-\lambda)(K-2)$ when  $\lambda \in [1/2,1]$. In this section we improve these bounds, in particular, for $\lambda \in (0,\frac{1}{2})$.
 Let us first rewrite~(\ref{owahc}) as follows
 \begin{align*}
 \mathrm{OWA}_{\pmb{w}}(\pmb{F}(\pmb{x}))) &= \lambda \max_{k\in[K]} \pmb{c}^T_k \pmb{x} + (1-\lambda) \min_{i\in[K]} \pmb{c}^T_i \pmb{x} = \min_{i\in[K]} \left[ \lambda\max_{k \in[K]} \pmb{c}^T_k \pmb{x} + (1-\lambda)\pmb{c}^T_i \pmb{x} \right]\\
&= \min_{i\in[K]} \left[\max_{k\in[K]}\  \left( \lambda\pmb{c}^T_k \pmb{x} + (1-\lambda)\pmb{c}^T_i \pmb{x}\right)\right]
=\min_{i\in [K]} \left[\max_{k\in[K]} \left(\lambda\pmb{c}_k+(1-\lambda)\pmb{c}_i \right)^T \pmb{x}\right].
\end{align*}
Accordingly,  the $\textsc{Hurwicz}~\mathcal{P}$ problem is equivalent  to solving 
$K$ $\textsc{Min-Max}~\mathcal{P}$ subproblems, i.e. 
$\min_{\pmb{x}\in  \X} \max_{k\in[K]} \left(\lambda\pmb{c}_k+(1-\lambda)\pmb{c}_i \right)^T \pmb{x}$
for every $i \in [K]$, and
choosing a solution that belongs to the best-performing subproblem -- this is a key fact.
The above equivalence leads, among others, to the following theorem.
\begin{theorem}
	If  $\textsc{Min-Max}~\mathcal{P}$ is approximable within $\alpha>1$
	(for $\alpha=1$ it is
	 polynomially solvable), then $\textsc{Hurwicz}~\mathcal{P}$ is approximable within  $\alpha$.
	\label{thmhur}
\end{theorem}
Therefore any $\alpha$-approximation algorithm for  $\textsc{Min-Max}~\mathcal{P}$
 can be used to find  an $\alpha$-approximate solution for $\textsc{Hurwicz}~\mathcal{P}$.
 
A comparison of our approach, which calls as a subroutine a general $K$-approximation algorithm 
 for each min-max subproblem
$\min_{\pmb{x}\in  \X} \max_{k\in[K]} \left(\lambda\pmb{c}_k+(1-\lambda)\pmb{c}_i \right)^T \pmb{x}$,
$i \in [K]$, 
with the approximation results shown in~\cite{KZ15} is depicted  in Figure~\ref{figH}. 
Here, we assume a $K$-approximation algorithm (see, e.g.,~\cite{ABV09}), but  even stronger 
approximation algorithms exist for particular
min-max problems~\cite{AFK02,D13,KMU08,KZ11}.
It is evident that our new approach is better for $\lambda\in (0,\frac{1}{2})$. 
\begin{figure}[htbp]
\centering
\includegraphics[width=0.5\textwidth]{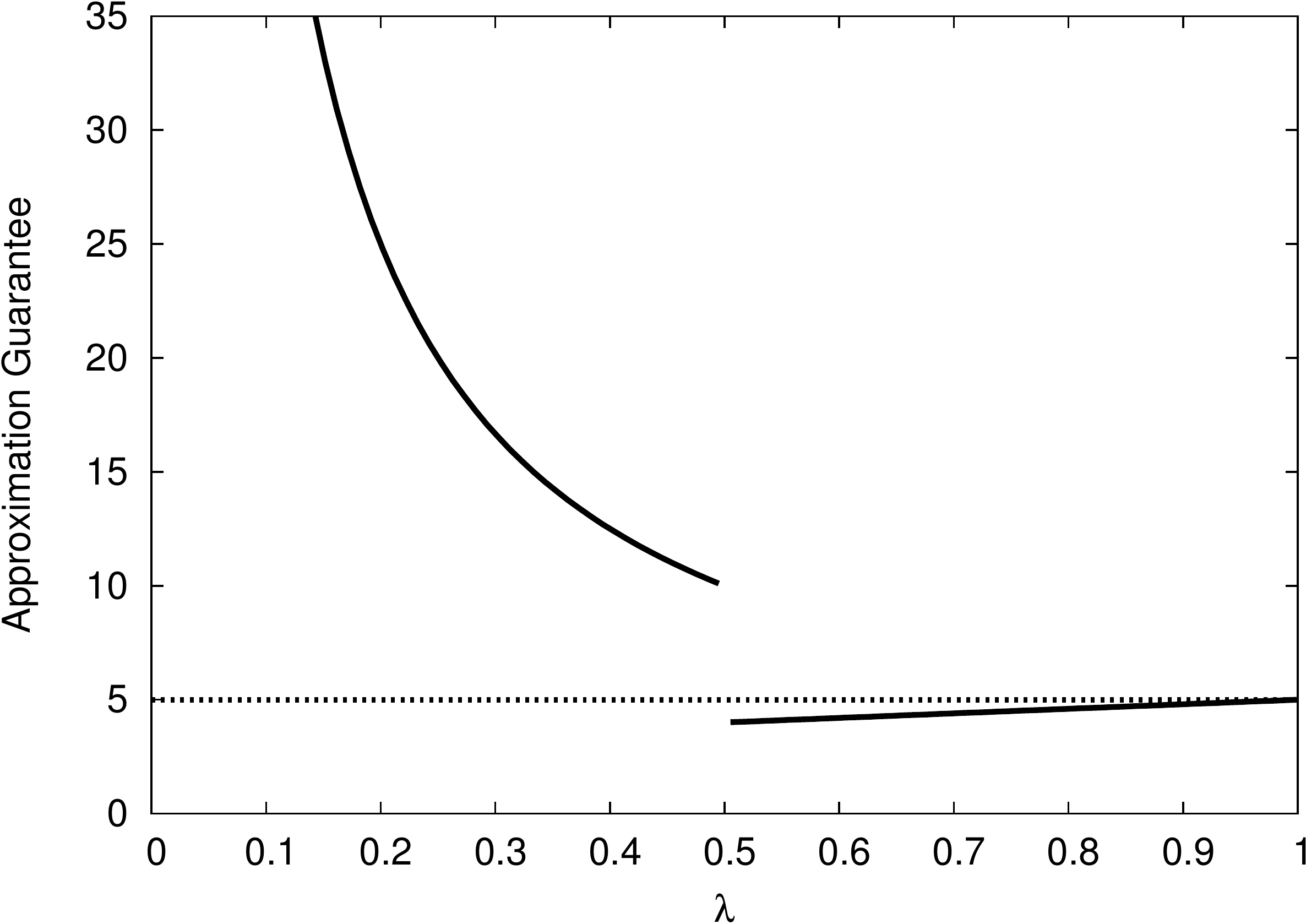}
\caption{Approximation guarantees for  $\textsc{Hurwicz}~\mathcal{P}$  with $K=5$. The solid line represents the 
approximation results from \cite{KZ15}. The dashed line is a constant guarantee $K$.}\label{figH}
\end{figure}

It is worth pointing out  that each min-max subproblem
is $\textsc{OWA}~\mathcal{P}$ with nonincreasing weights
($w_1=1$ and $w_k=0$ for $k>1$),
so the results obtained in Section~\ref{sec2} ($\ell$-\textsc{Aggregation Algorithm}) or in~\cite{CG18}
can be also applied to  $\textsc{Min-Max}~\mathcal{P}$ and, in consequence, to  $\textsc{Hurwicz}~\mathcal{P}$.

\section{Computational tests}
\label{sec4}

In this section, we present computational experiments illustrating 
 the practical performance of solutions to $\textsc{OWA}~\mathcal{P}$
  that are found through our objective aggregation approach.

\subsection{Setup}
To test the practical performance of our aggregation algorithm, we have chosen a \textsc{Min-Knapsack} problem
 of the form (see, e.g.~\cite{WS10}):
\[ 
\textsc{Min-Knapsack: }  \min_{\pmb{x}\in\X} \pmb{c}^T\pmb{x} \qquad \text{ with } \qquad \X = \{\pmb{x}\in\{0,1\}^n : \pmb{b}^T\pmb{x} \ge B\},
\]
where $c_i, b_i\geq  0$, $i\in [n]$, and $B>0$ are given.
To solve the resulting \textsc{OWA Min-Knapsack} problem, we reformulate it using the technique from \cite{CG15} to find
\begin{align*}
\textsc{OWA Min-Knapsack: }\min\ & \sum_{k\in[K]} \pi_k + \rho_k \\
\text{s.t. } & \sum_{i\in[n]} b_ix_i \ge B \\
& \pi_k + \rho_j \ge \sum_{i\in[n]} w_k c_{ij} x_i & \forall j,k\in[K] \\
& x_i \in \{0,1\} & \forall i\in[n]
\end{align*}
We generate several problem sets with different parameters. We set $n=40$ for all experiments, and consider instances with $K=50$ and $K=200$. We generate item weights $b_i$, $i\in [n]$, by sampling i.i.d. uniformly from the interval $[0.1,10]$. We set $B= 1/3\sum_{i\in[n]} b_i$. Additionally, we test two methods to generate objectives $\pmb{c}$, and two methods to generate weights $\pmb{w}$.

In the first method to generate item costs $c_{ik}$, $i\in[n], k\in [K]$, we sample i.i.d. uniformly from $[0.5,1.5]$ and multiply this number with $b_i$ (i.e., item weights and costs are correlated). In the second method, we assume that objective functions have more structure. We generate $K'<K$ nominal scenarios in the same way as for the first method. We then sample $K$ scenarios, by first choosing a random nominal scenario, and then multiplying all costs of this scenario with random values sampled i.i.d. uniformly from $[0.8,1.2]$.

For the weight vectors $\pmb{w}$ with $w_1\geq w_2\geq\cdots \geq w_K$, the first method uses the following 
generating function (see~\cite{KZ16}):
\begin{align}
g_{\alpha}(z)&=\frac{1}{1-\alpha}(1-\alpha^z)  &\label{genfun}\\
w_k&=g_{\alpha}\left(\frac{k}{K}\right)-g_\alpha \left(\frac{k-1}{K}\right)&k\in [K], \nonumber
\end{align}
where $\alpha\in (0,1)$ is a fixed parameter.
It is easily seen that the greater the value of~$\alpha$, the less distorted is the weight distribution (i.e. it is closer to uniform).
The second method to generate weight vectors uses the $p$-centra setting, where for a fixed $p\in[K]$, $w_i = 1/p$ for $i=1,\ldots,p $, and $w_i=0$ for all other $i$.

All parameter settings and the corresponding instance names are summarized in Table~\ref{exp-table}. Each experiment is repeated 200 times, and results are averaged.

\begin{table}[htb]
\begin{center}
\begin{tabular}{r|rrrr}
Name & $n$ & $K$ & $\pmb{c}$ & $\pmb{w}$ \\
\hline
$\mathcal{I}^1_1$ & 40 & 50 & uni & $\alpha = 10^{-1}$ \\
$\mathcal{I}^1_2$ & 40 & 50 & uni & $\alpha = 10^{-3}$ \\
$\mathcal{I}^1_3$ & 40 & 50 & uni & $p = 0.1K$ \\
$\mathcal{I}^1_4$ & 40 & 50 & uni & $p = 0.3K$ \\
\hline
$\mathcal{I}^2_1$ & 40 & 50 & $K'=10$ & $\alpha = 10^{-1}$ \\
$\mathcal{I}^2_2$ & 40 & 50 & $K'=10$ & $\alpha = 10^{-3}$ \\
$\mathcal{I}^2_3$ & 40 & 50 & $K'=10$ & $p = 0.1K$ \\
$\mathcal{I}^2_4$ & 40 & 50 & $K'=10$ & $p = 0.3K$ \\
\hline
$\mathcal{J}^1_1$ & 40 & 200 & uni & $\alpha = 10^{-1}$ \\
$\mathcal{J}^1_2$ & 40 & 200 & uni & $\alpha = 10^{-3}$ \\
$\mathcal{J}^1_3$ & 40 & 200 & uni & $p = 0.1K$ \\
$\mathcal{J}^1_4$ & 40 & 200 & uni & $p = 0.3K$ \\
\hline
$\mathcal{J}^2_1$ & 40 & 200 & $K'=10$ & $\alpha = 10^{-1}$ \\
$\mathcal{J}^2_2$ & 40 & 200 & $K'=10$ & $\alpha = 10^{-3}$ \\
$\mathcal{J}^2_3$ & 40 & 200 & $K'=10$ & $p = 0.1K$ \\
$\mathcal{J}^2_4$ & 40 & 200 & $K'=10$ & $p = 0.3K$ \\
\end{tabular}
\caption{Problem instances with parameter settings.}\label{exp-table}
\end{center}
\end{table}

All experiments were carried out on a 16-core Intel Xeon E5-2670 processor, running at 2.60 GHz with 20MB cache, and Ubuntu 12.04. Processes were pinned to one core. We used CPLEX v.12.6 to solve all problem formulations with a timelimit of 60 seconds.

We compare two algorithms in this setting. The first algorithm is the $\ell$-\textsc{Aggregation Algorithm} from Section~\ref{sec2-2}. As the second algorithm, we also consider the aggregation based on using $\overline{K}$-means from Section~\ref{sec2-3}\footnote{To solve the clustering problem, we used C++ library by John Burkardt from \url{http://people.sc.fsu.edu/~jburkardt/cpp_src/kmeans/kmeans.html}}. In the following, for brevity, the first approach is referred to as Alg1, and the latter as Alg2.

\subsection{Results}

We first present averaged results over all instance types with $K=50$ and all instance types with $K=200$, respectively. Figure~\ref{figexp-p} shows average objective values, while Figure~\ref{figexp-time} presents computation times.

\begin{figure}[htbp]
\centering
\subfigure[$K=50$.\label{figexp-p50}]{\includegraphics[width=0.45\textwidth]{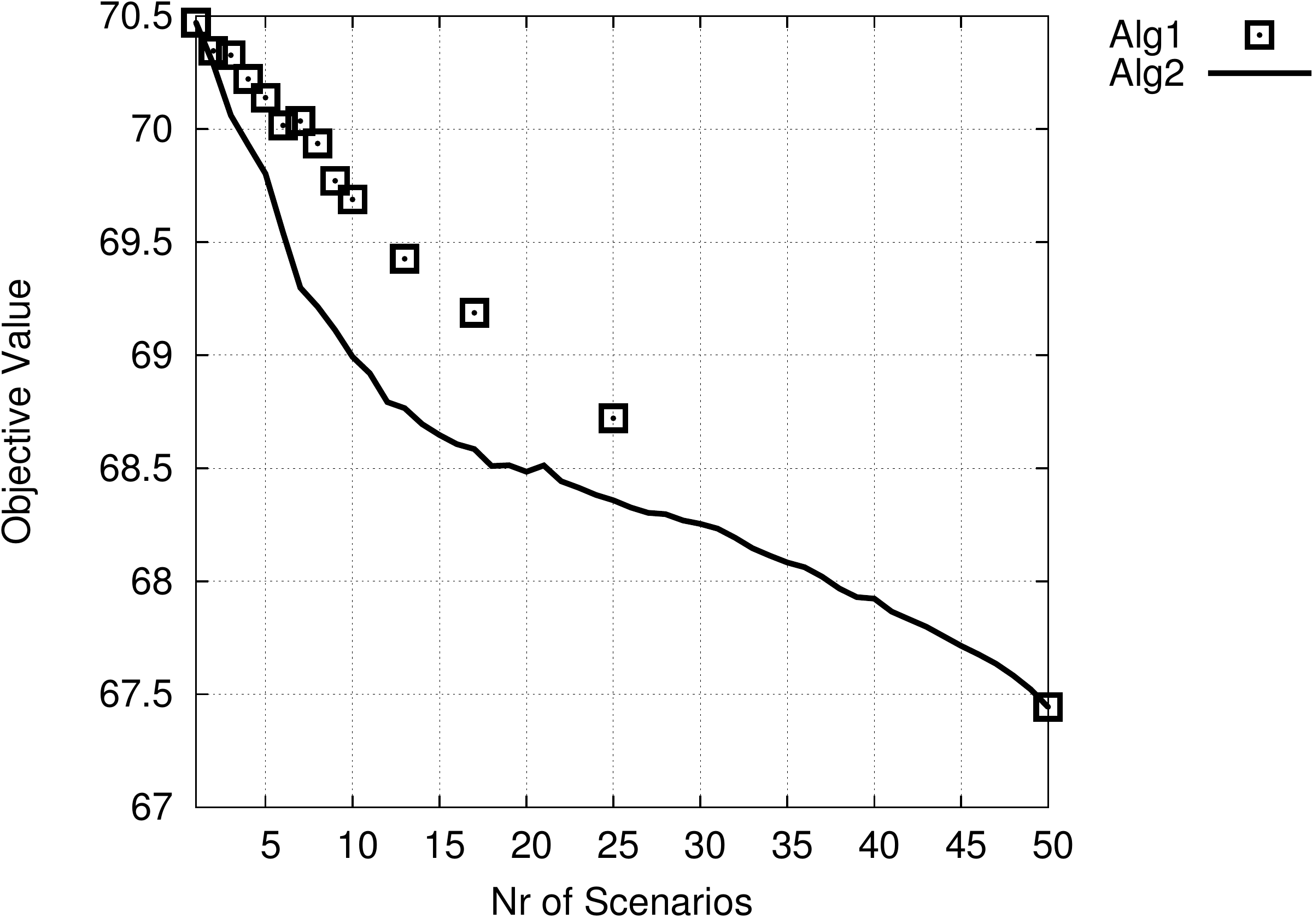}}
\subfigure[$K=200$.\label{figexp-p200}]{\includegraphics[width=0.45\textwidth]{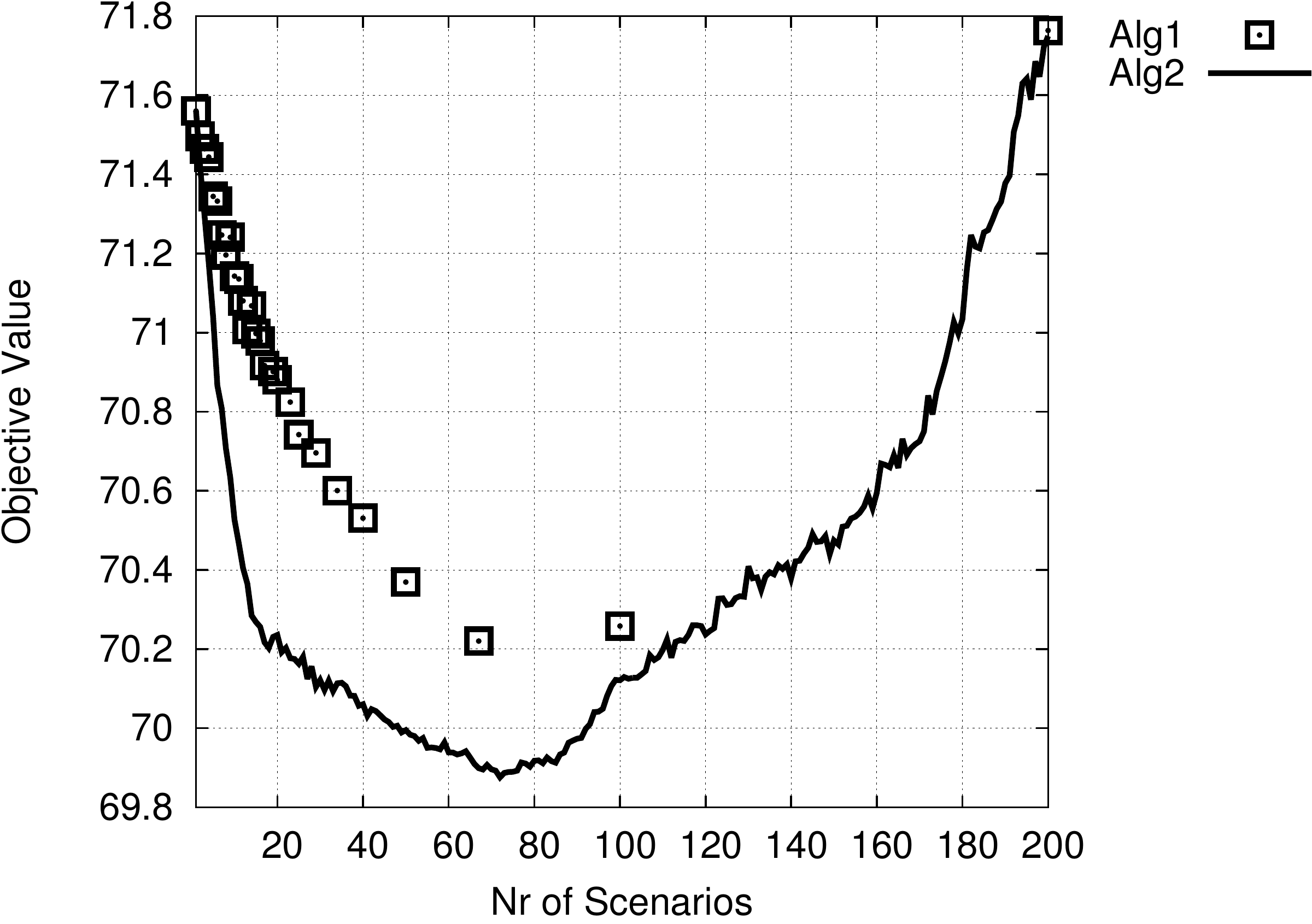}}
\caption{Average objective values of best solutions found for \textsc{OWA Min-Knapsack}.} \label{figexp-p}
\end{figure}

\begin{figure}[htbp]
\centering
\subfigure[$K=50$.\label{figexp-time50}]{\includegraphics[width=0.45\textwidth]{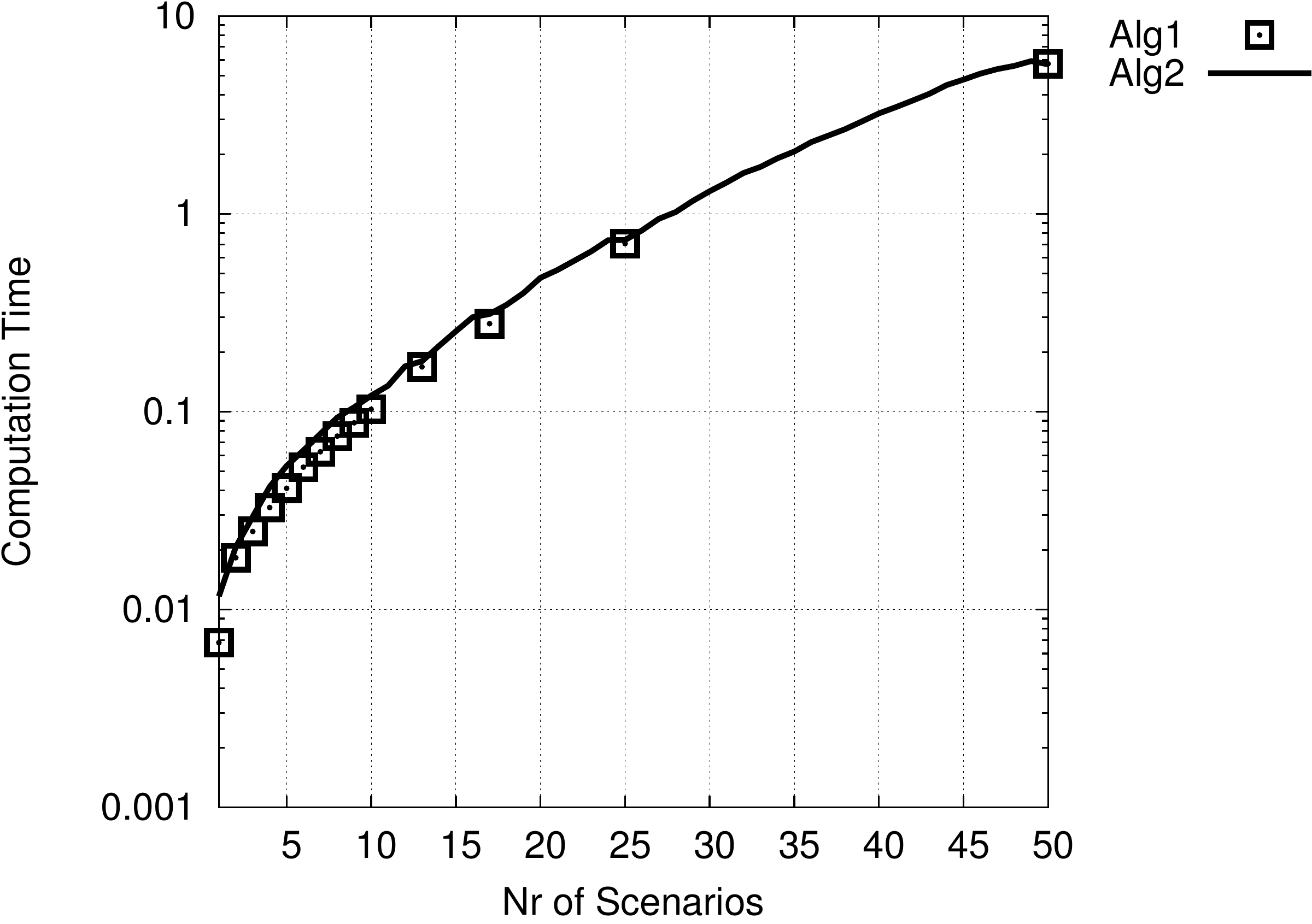}}
\subfigure[$K=200$.\label{figexp-time200}]{\includegraphics[width=0.45\textwidth]{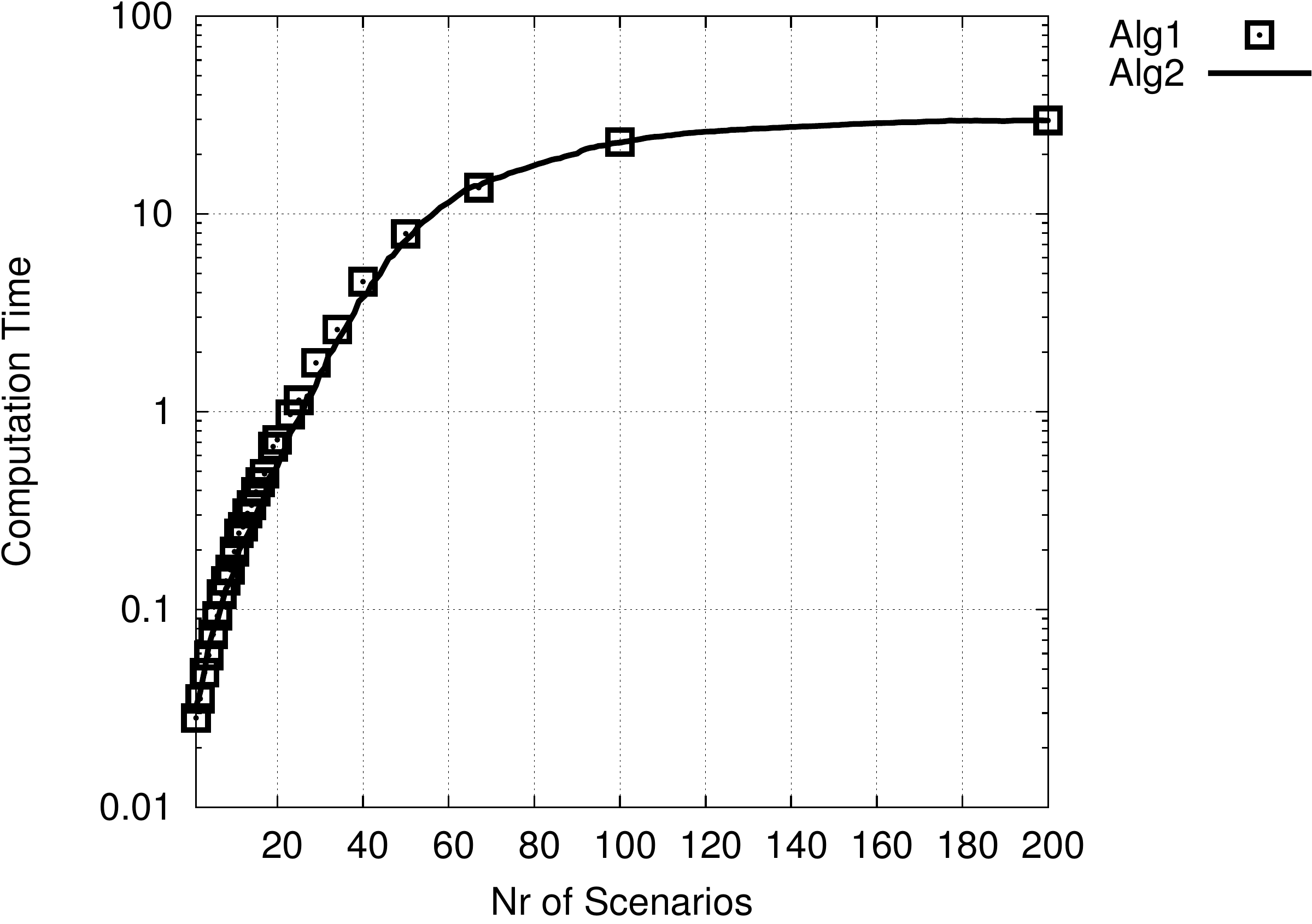}}
\caption{Average computation times for \textsc{OWA Min-Knapsack}.} \label{figexp-time}
\end{figure}

On all the figures, the horizontal axis shows how many scenarios (objectives functions) were left after aggregation. This means that for the leftmost point, we use an average scenario, while for the rightmost point, we solve the original \textsc{OWA Min-Knapsack} problem. Note that Alg1 is presented by discrete points, which is due to the fact that different choices of $\ell$ can result in the same reduced problem size. For example, for $\ell=2$, $K=50$ objectives are aggregated down to 25 objectives. Therefore, all values between 25 and 50 would correspond to the same solution. Also note that using $\ell$ can lead to the use of dummy scenarios.

For the case $K=50$, all instances were solved to optimality. Figure~\ref{figexp-p50} shows that the ratio between the objective values of heuristic and optimal solutions is much smaller than the theoretical bound indicated. In general, aggregations using more objectives can give better solutions than using less objectives. This is particularly the case for Alg2, which outperforms the naive aggregation Alg1. Figure~\ref{figexp-time50} shows that the computation time increases with the number of objectives that are used (note the logarithmic vertical scale). Also, the problems resulting from Alg1 and from Alg2 have the same difficulty.

The case for $K=200$ can be seen in Figure~\ref{figexp-p200}. Here, solving the original problem was not possible within the available computation time. We find that the average objective values of our aggregation methods are better than the average objective value of the exact approach. The figure shows a trade-off between using too few objectives and too many, with best results achieved using around 70 objectives.

In Figures~\ref{plots-p-0} and \ref{plots-p-10}, we show the average objective values for different instance types in more detail. For example, Figures~\ref{figI11} and~\ref{figJ11} show the average objective values for instances where item costs are generated uniformly and independently, and item weights are close to uniform. In this setting, Alg2 does not improve Alg1 for most aggregation levels. This is also the case for other sets where item weights are generated in this way. A possible explanation for this phenomenon is that in this case, using an aggregation where every aggregated objective results from the same number of original objectives (as is the case for Alg1) is beneficial. Note that Alg2 may find clusters of different size.

\begin{figure}[htbp]
\centering
\subfigure[$\mathcal{I}^1_1$.\label{figI11}]{\includegraphics[width=0.43\textwidth]{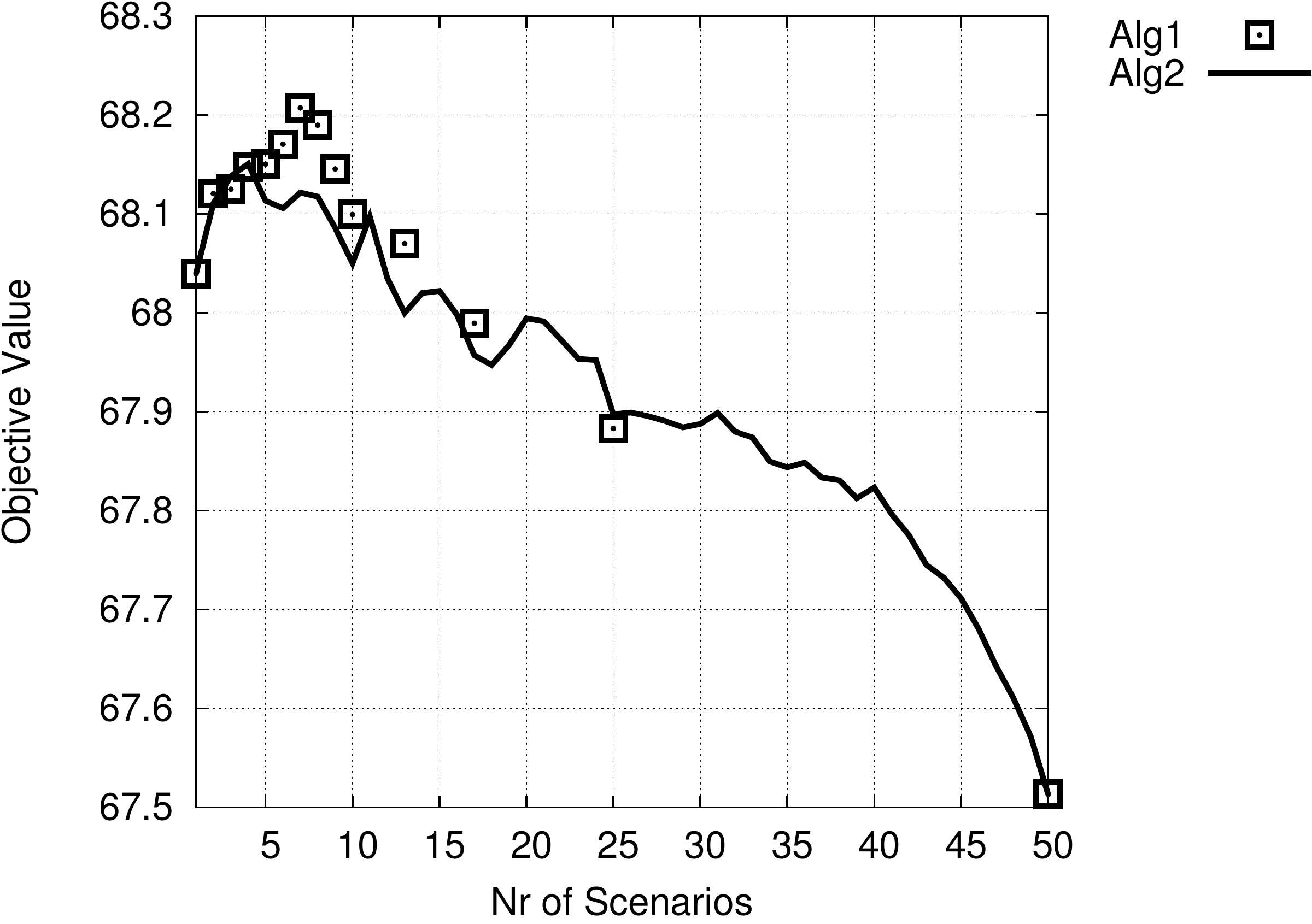}}
\subfigure[$\mathcal{J}^1_1$.\label{figJ11}]{\includegraphics[width=0.43\textwidth]{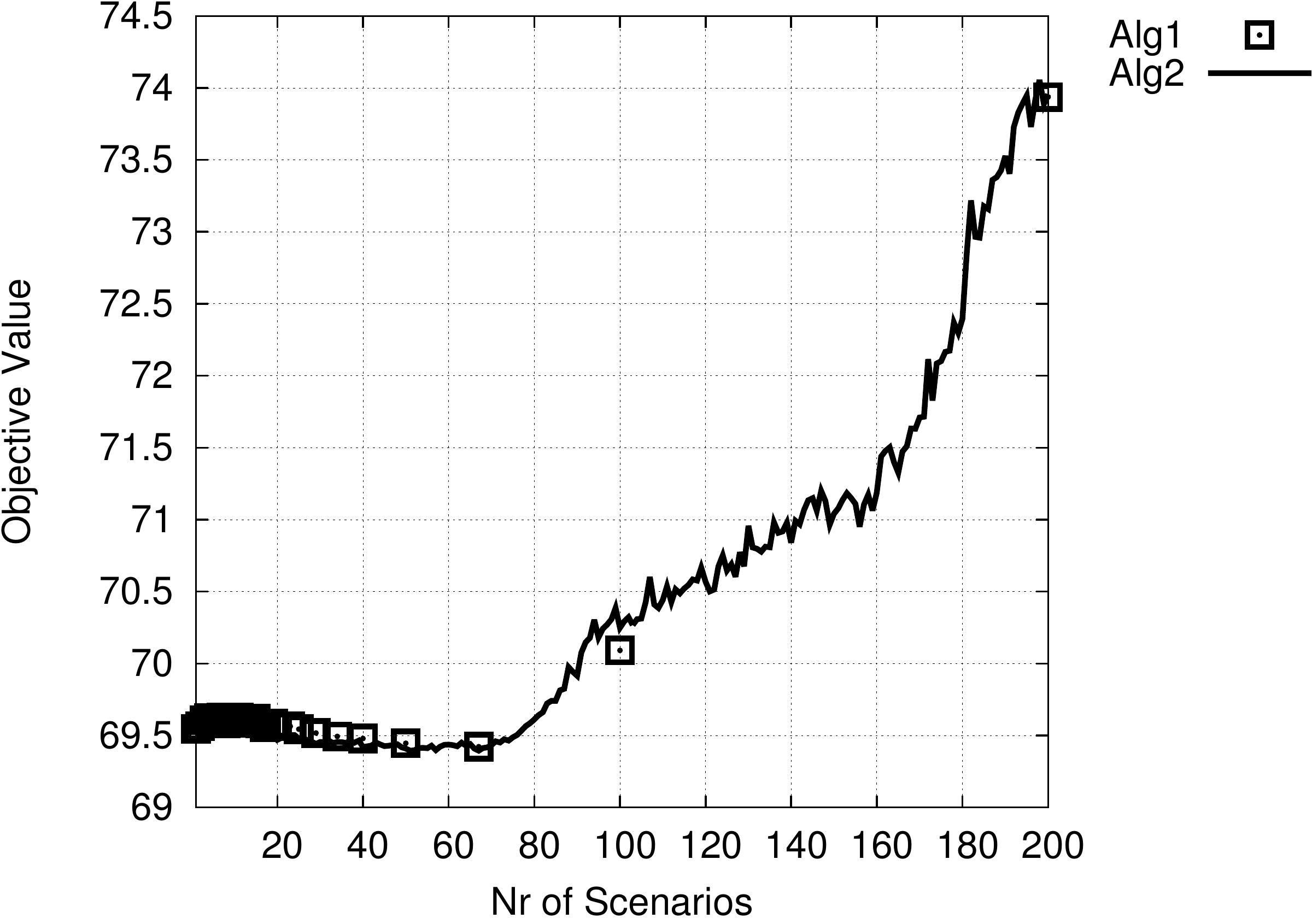}}
\subfigure[$\mathcal{I}^1_2$.]{\includegraphics[width=0.43\textwidth]{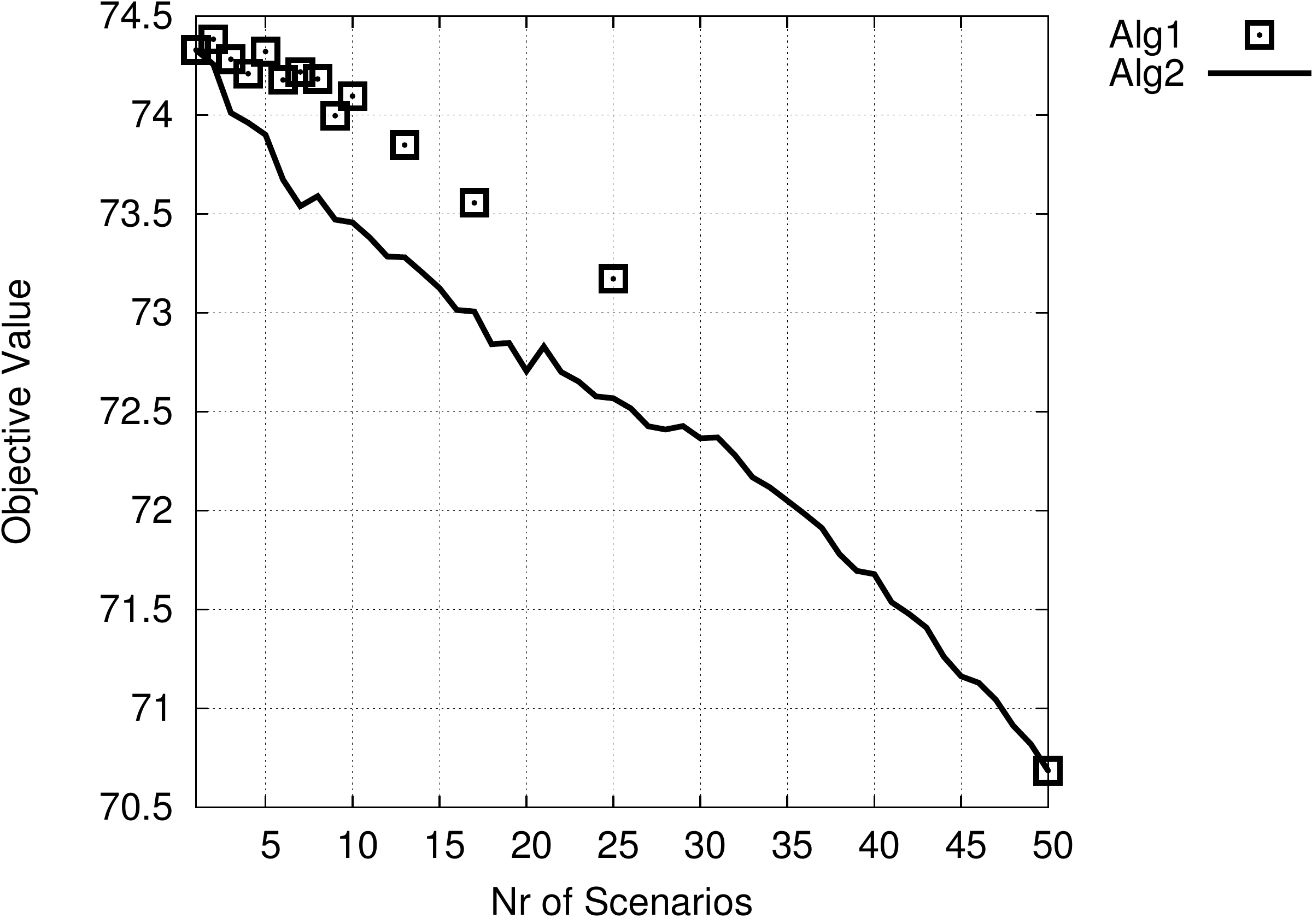}}
\subfigure[$\mathcal{J}^1_2$.]{\includegraphics[width=0.43\textwidth]{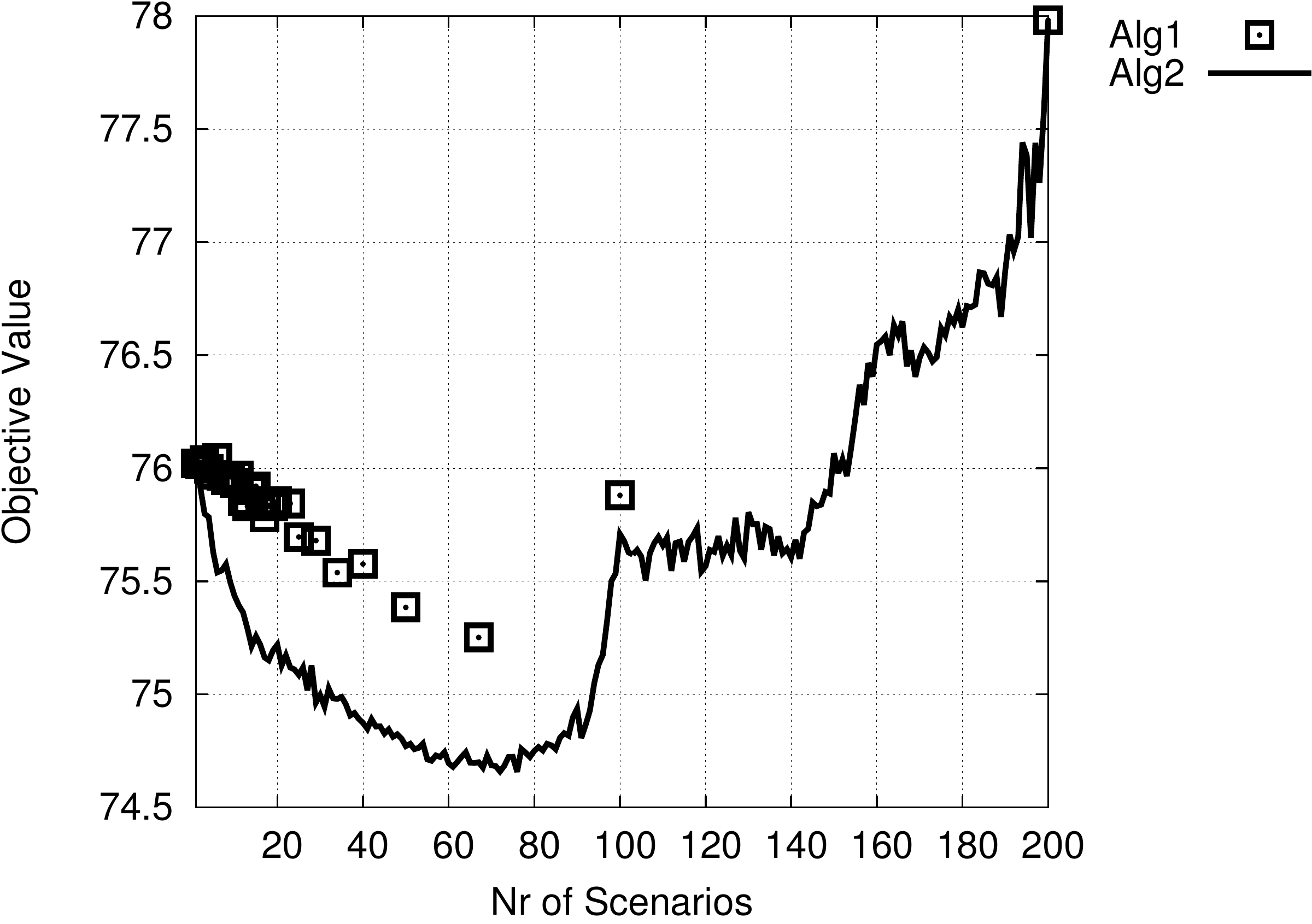}}
\subfigure[$\mathcal{I}^1_3$.]{\includegraphics[width=0.43\textwidth]{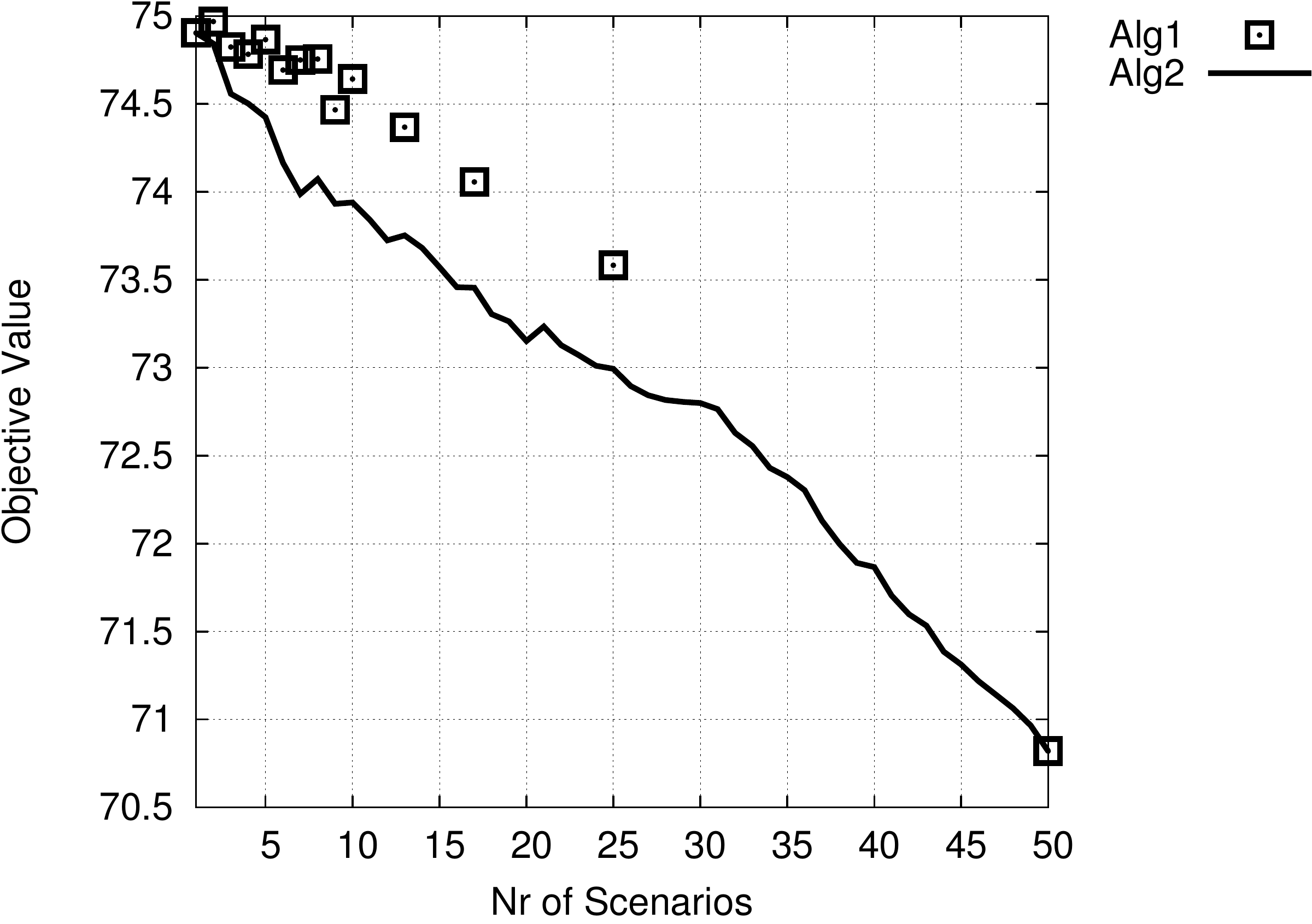}}
\subfigure[$\mathcal{J}^1_3$.]{\includegraphics[width=0.43\textwidth]{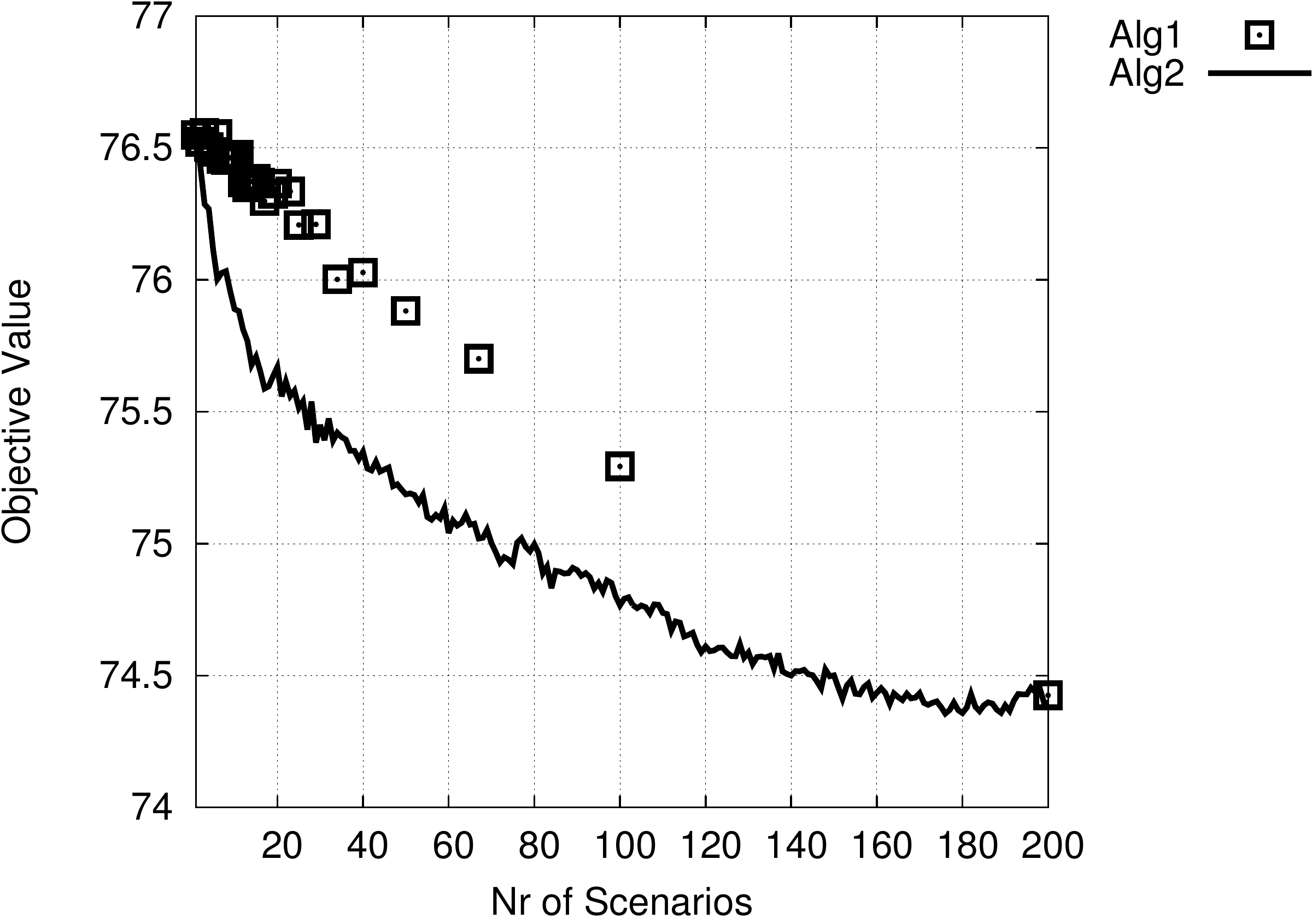}}
\subfigure[$\mathcal{I}^1_4$.]{\includegraphics[width=0.43\textwidth]{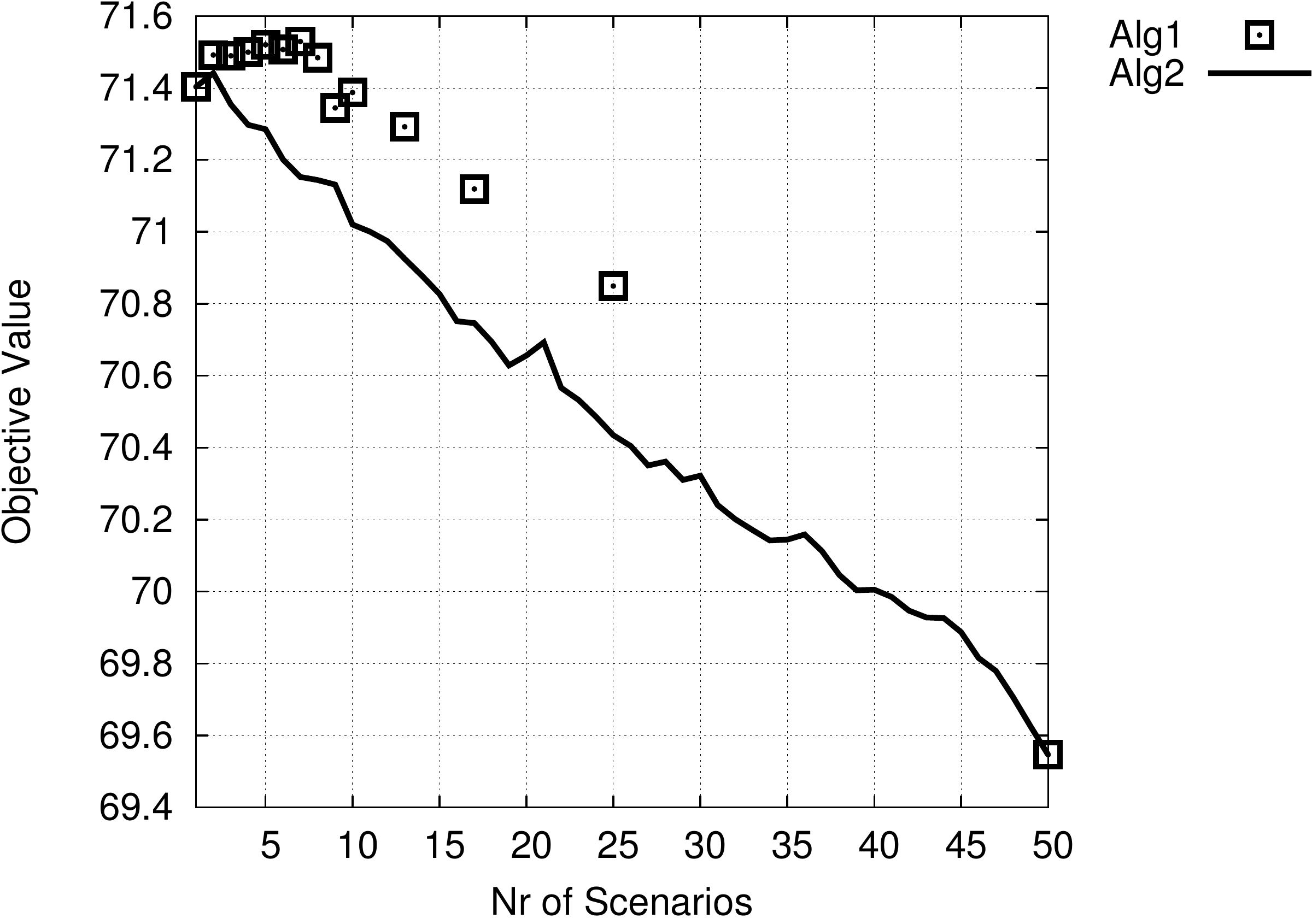}}
\subfigure[$\mathcal{J}^1_4$.]{\includegraphics[width=0.43\textwidth]{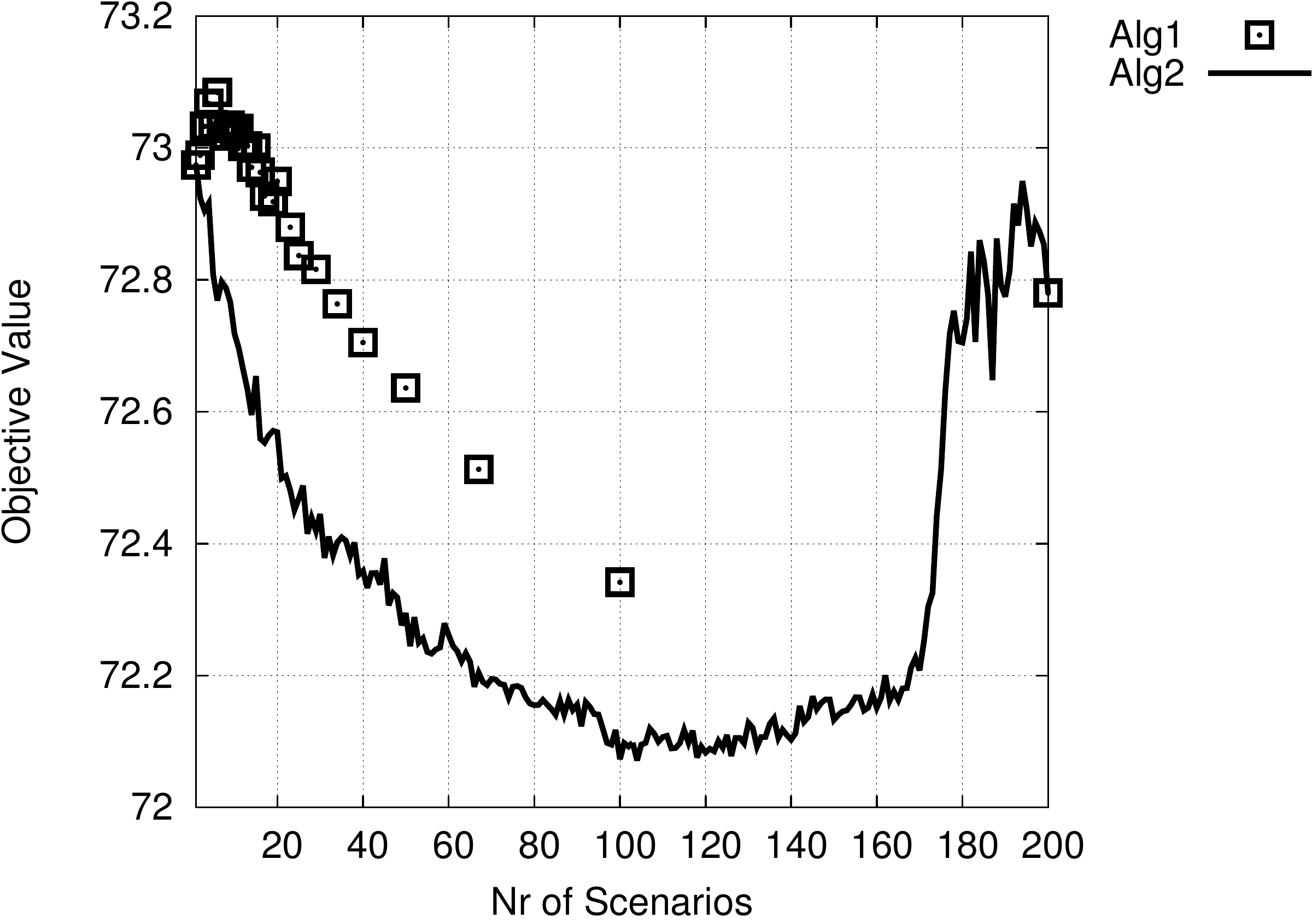}}
\caption{Average objective values for \textsc{OWA Min-Knapsack}.\label{plots-p-0}}
\end{figure}

\begin{figure}[htbp]
\centering
\subfigure[$\mathcal{I}^2_1$.]{\includegraphics[width=0.43\textwidth]{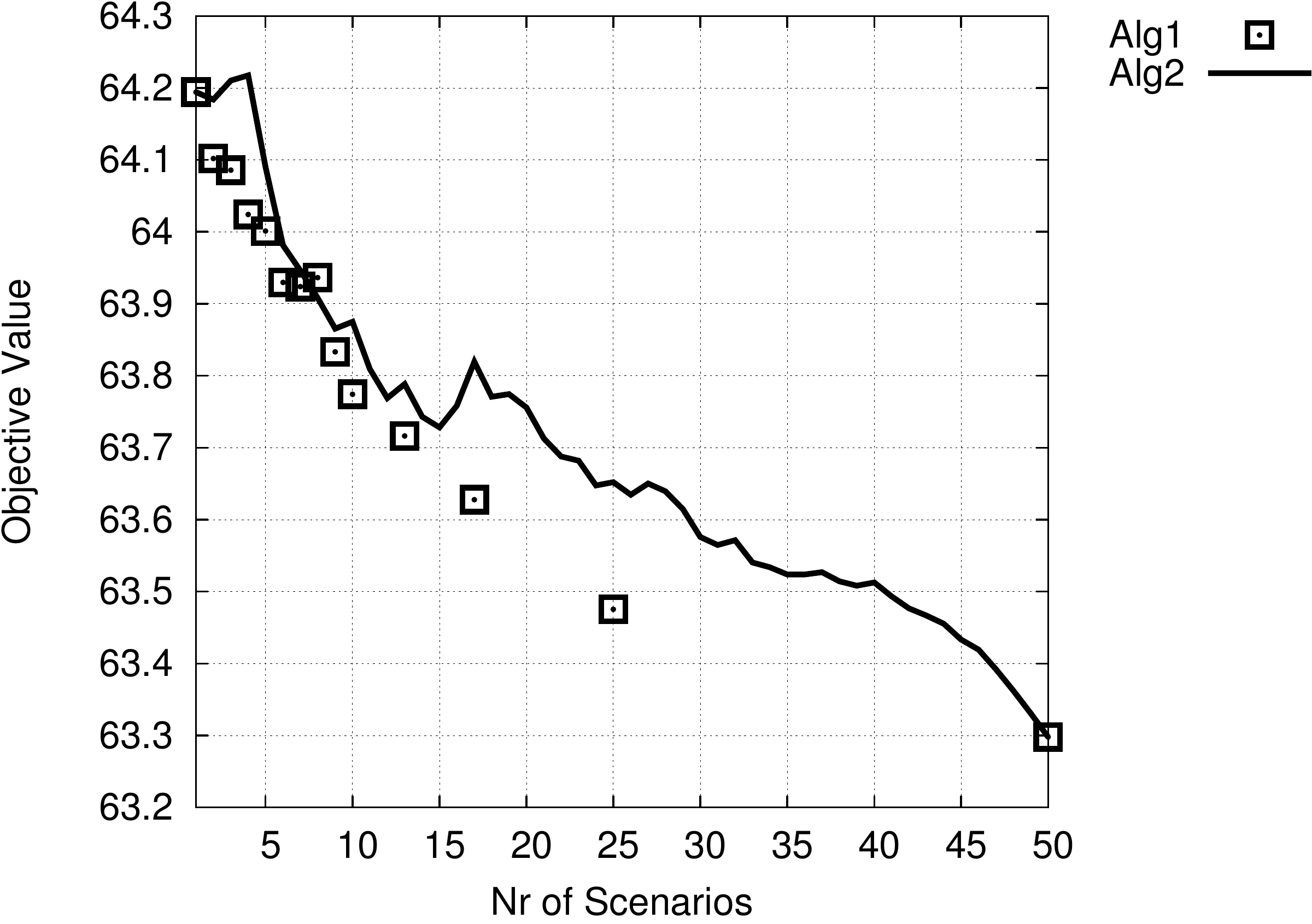}}
\subfigure[$\mathcal{J}^2_1$.]{\includegraphics[width=0.43\textwidth]{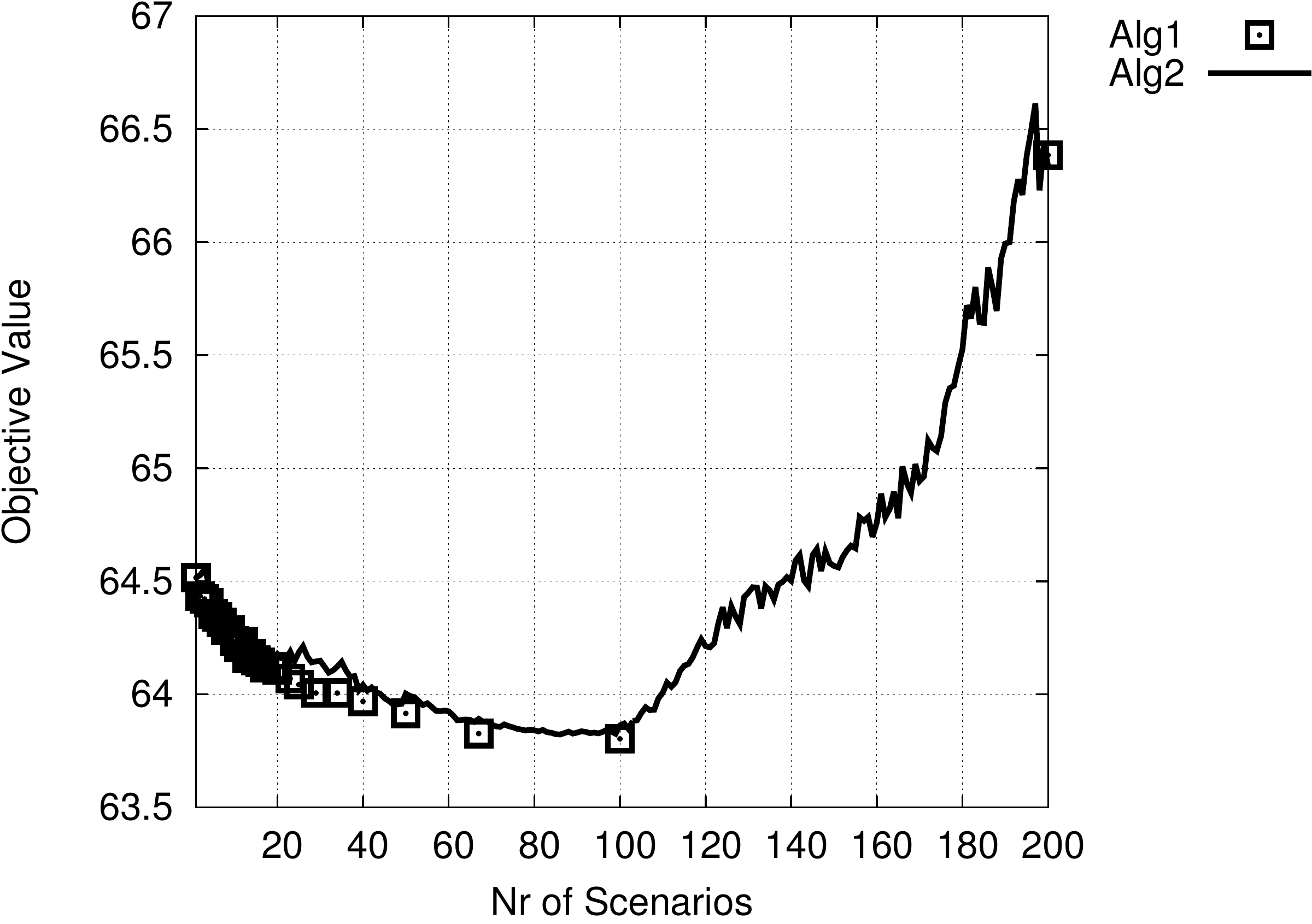}}
\subfigure[$\mathcal{I}^2_2$.\label{figI22}]{\includegraphics[width=0.43\textwidth]{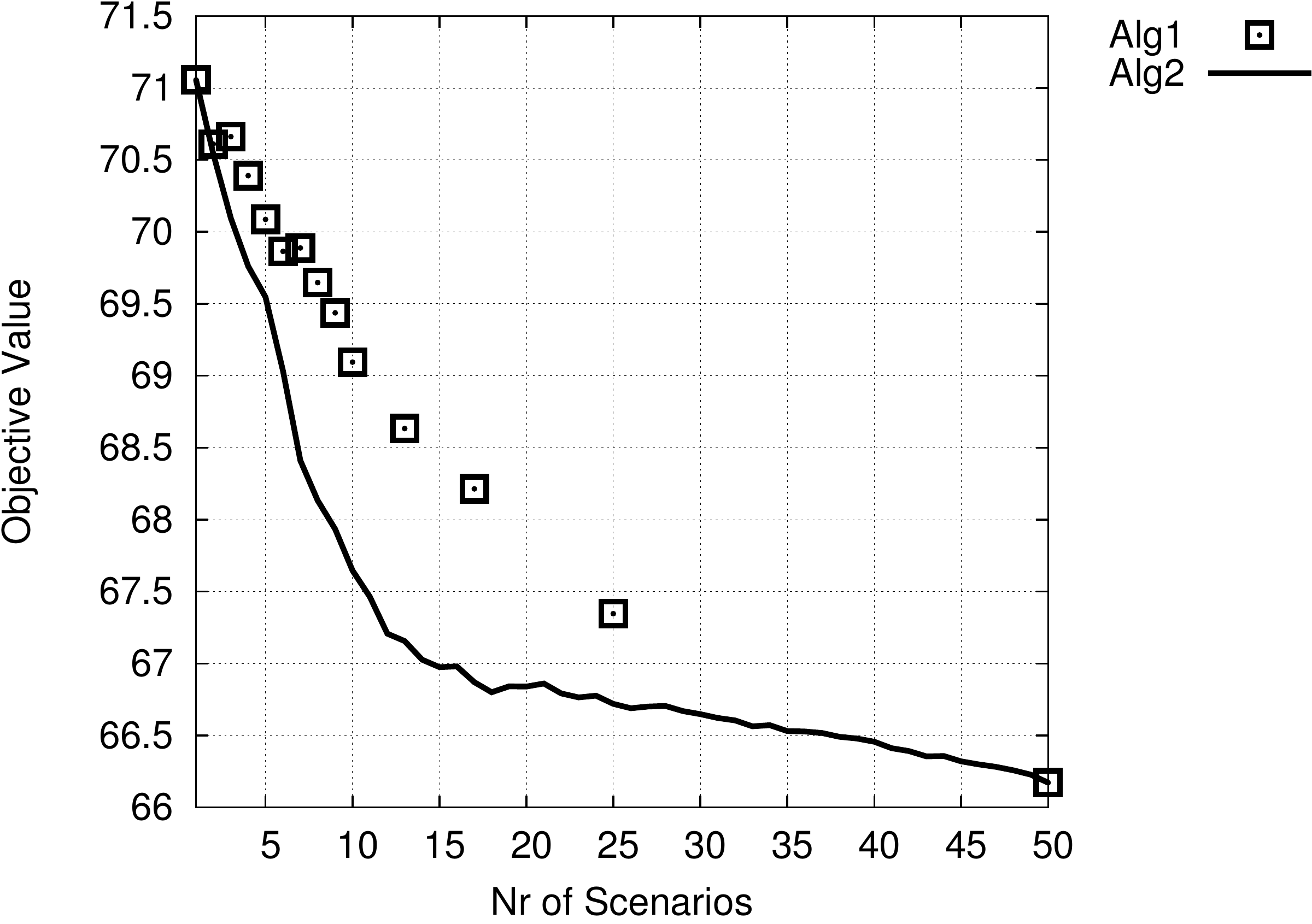}}
\subfigure[$\mathcal{J}^2_2$.]{\includegraphics[width=0.43\textwidth]{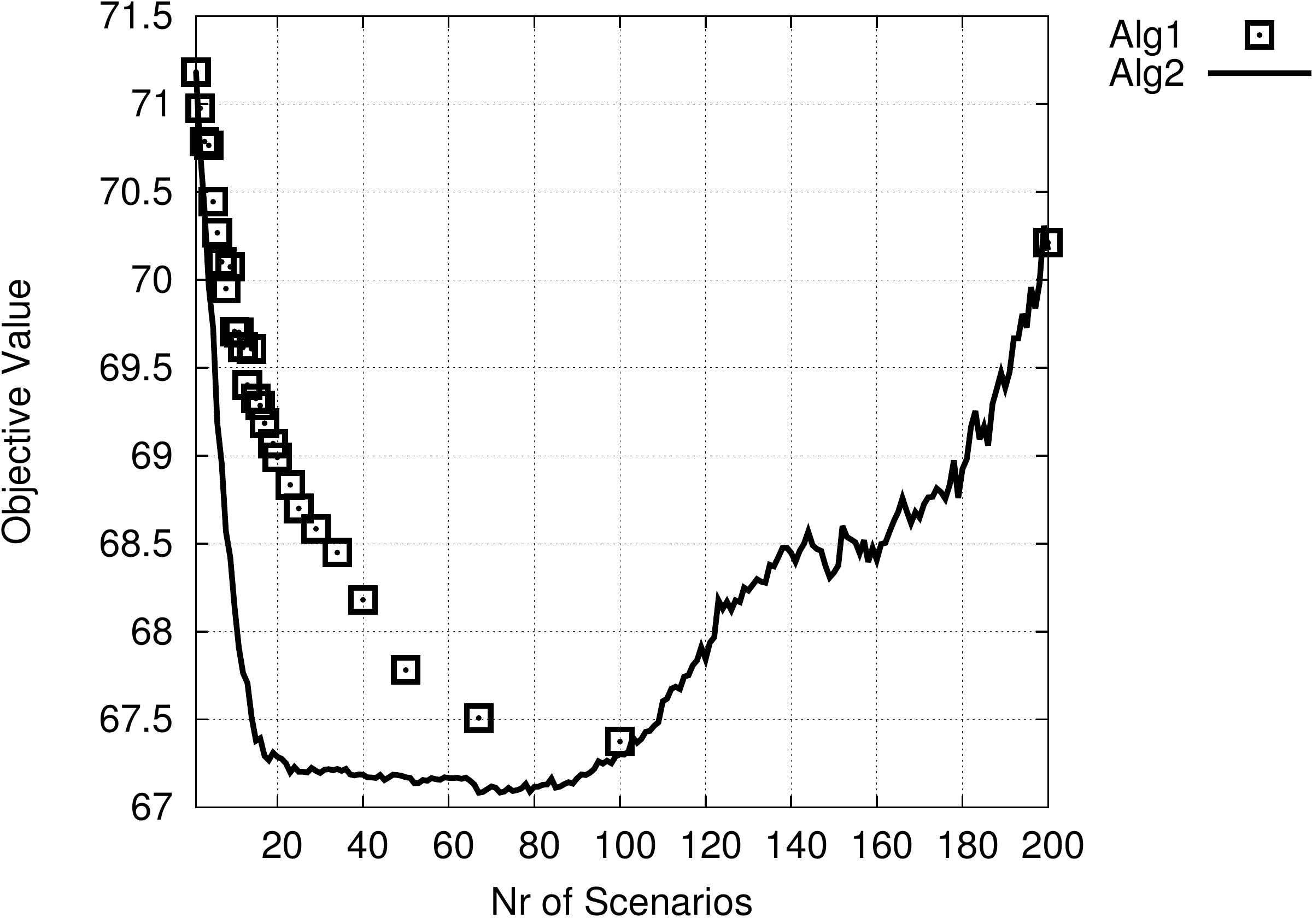}}
\subfigure[$\mathcal{I}^2_3$.]{\includegraphics[width=0.43\textwidth]{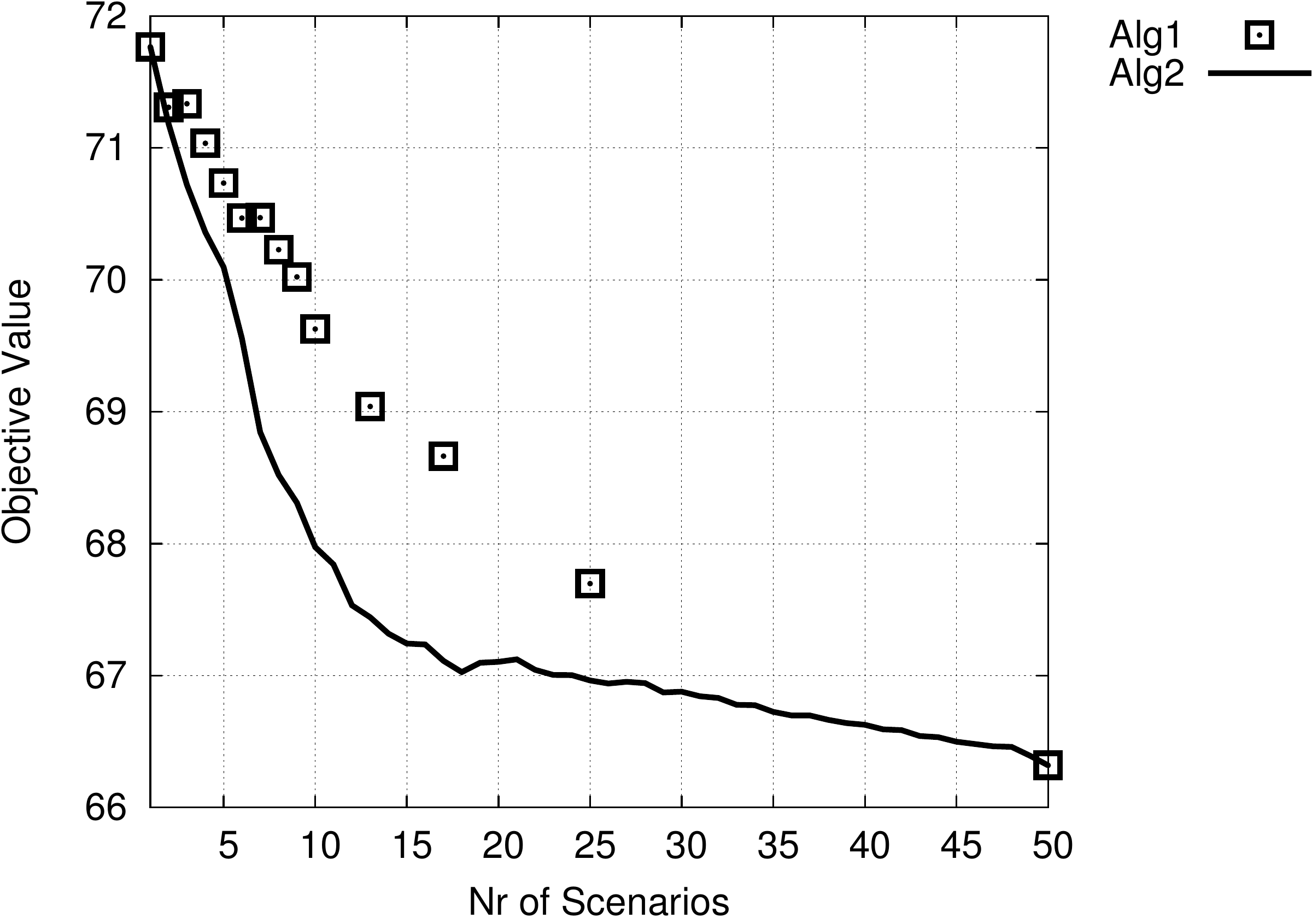}}
\subfigure[$\mathcal{J}^2_3$.]{\includegraphics[width=0.43\textwidth]{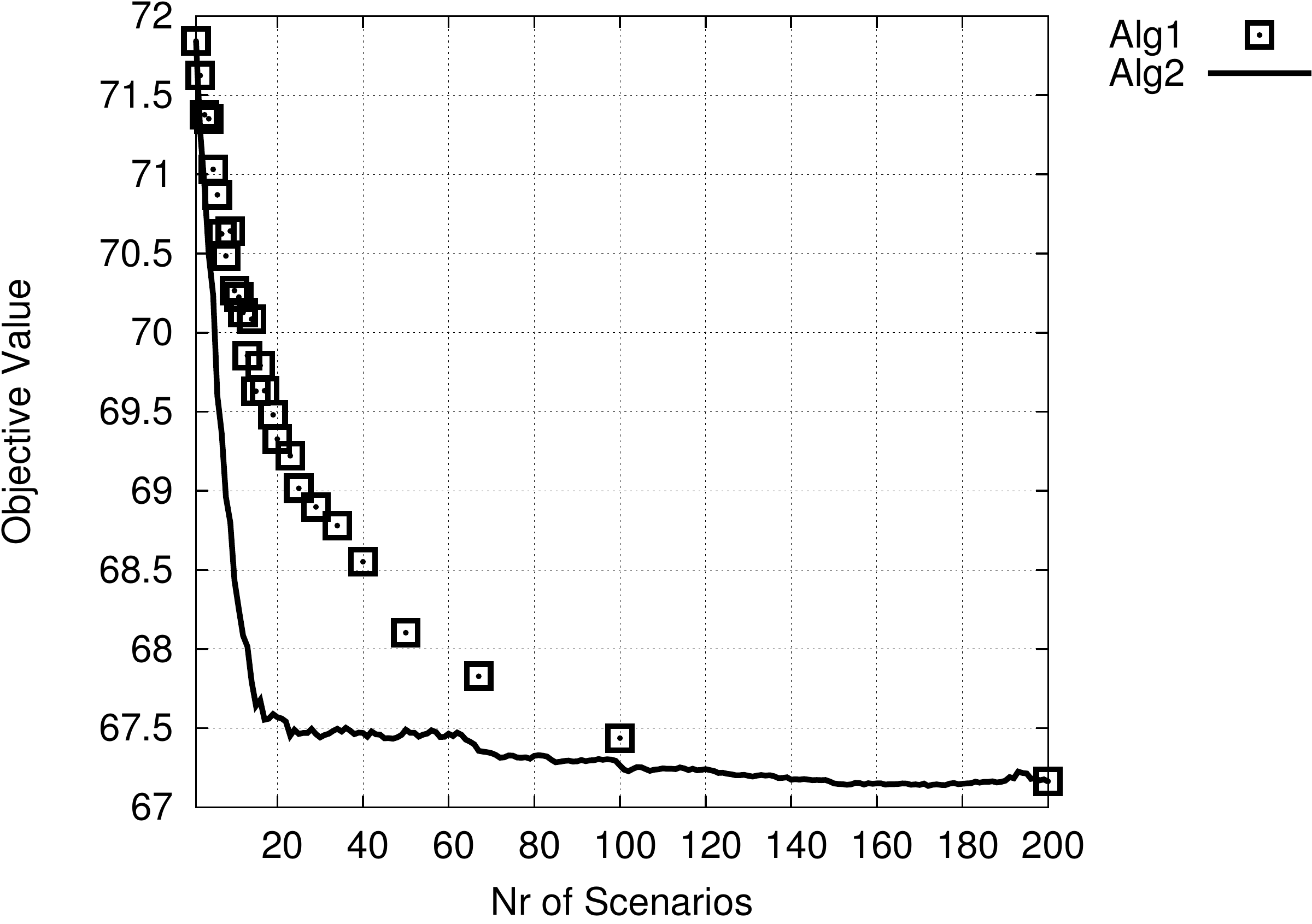}}
\subfigure[$\mathcal{I}^2_4$.]{\includegraphics[width=0.43\textwidth]{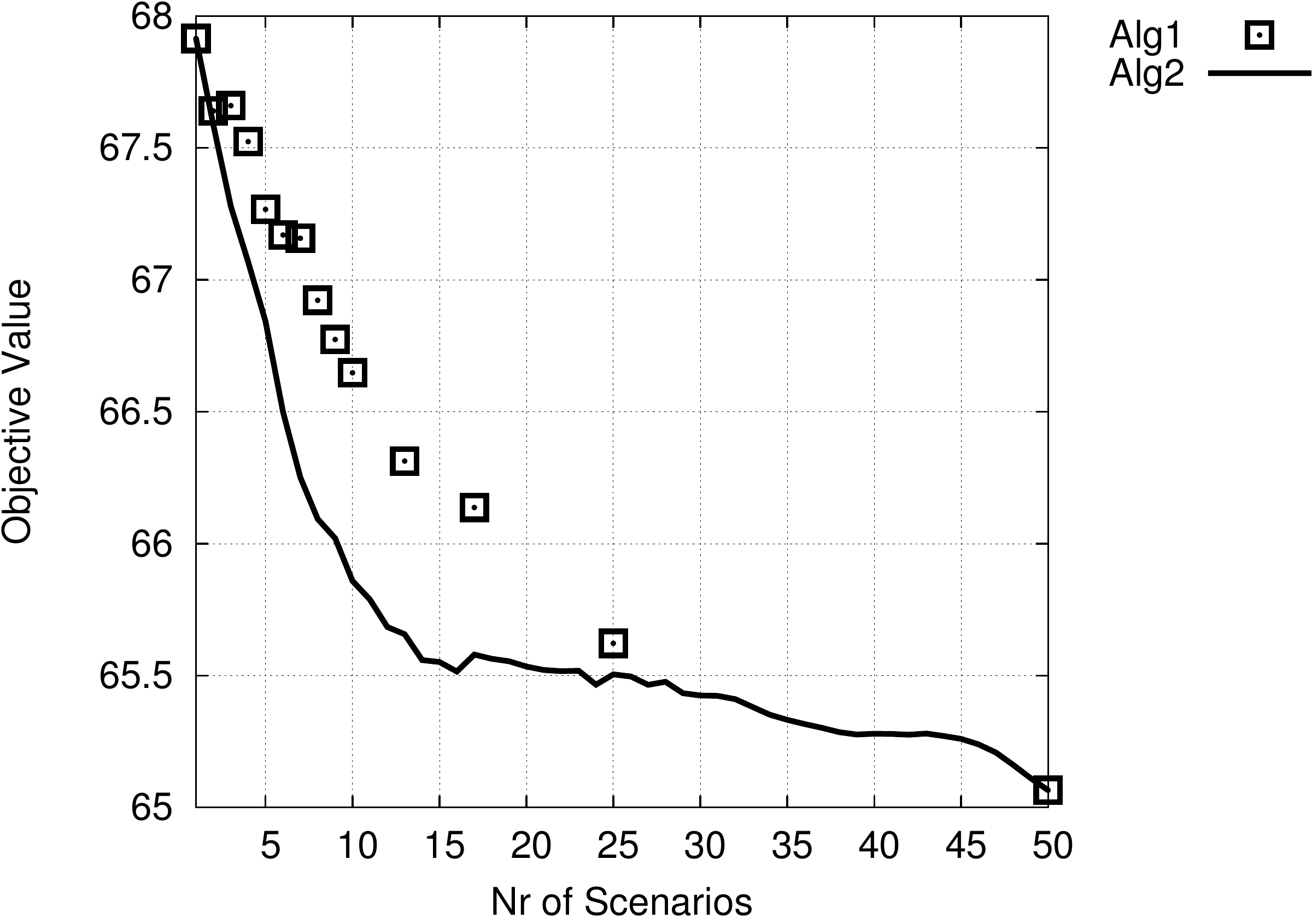}}
\subfigure[$\mathcal{J}^2_4$.]{\includegraphics[width=0.43\textwidth]{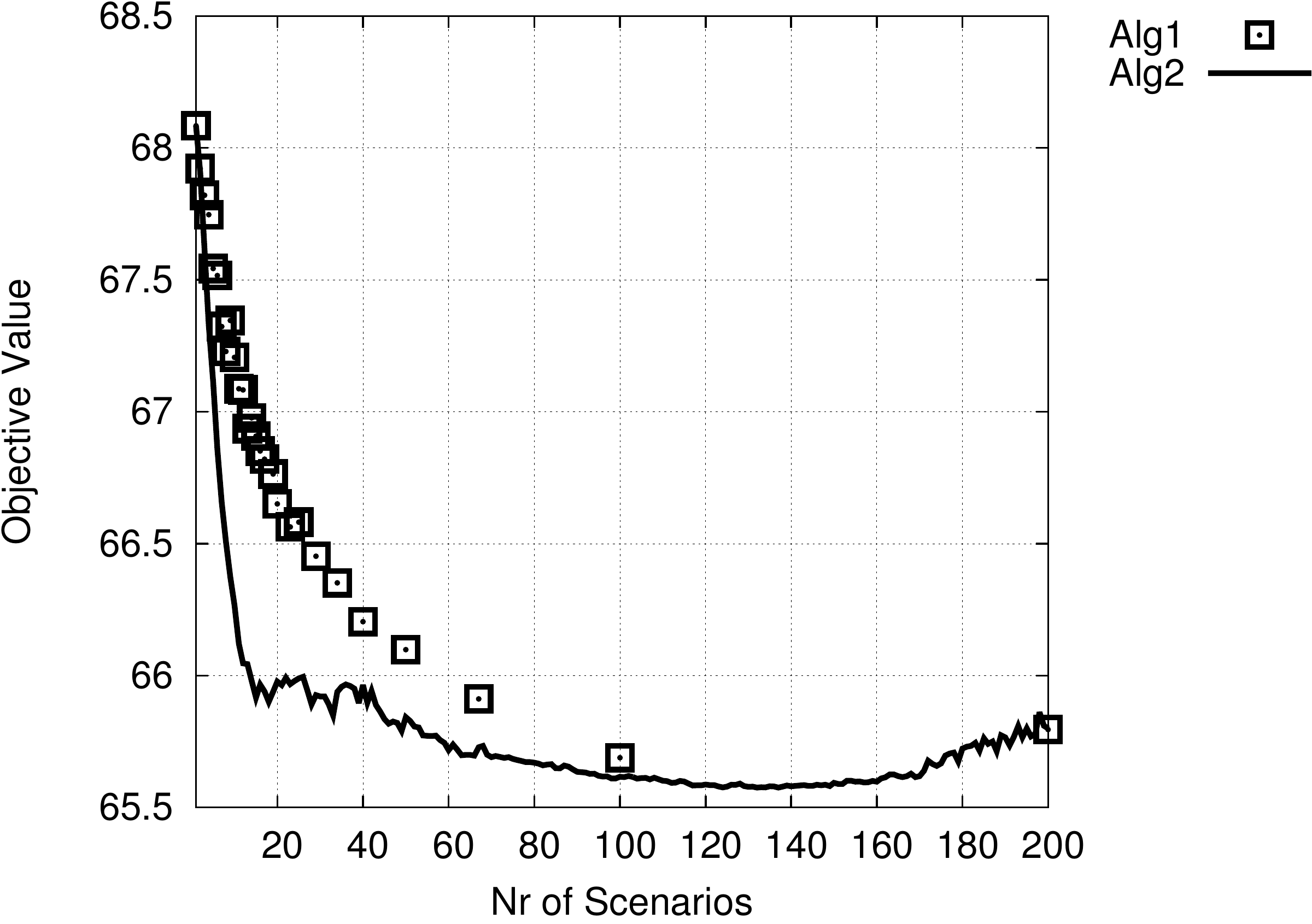}}
\caption{Average objective values for \textsc{OWA Min-Knapsack}.\label{plots-p-10}}
\end{figure}

Another particularity can be seen in Figure~\ref{figI22}, where $K'=10$ generating (nominal) objectives where used. The objective value of Alg2 decreases especially fast until around 10 objectives, where the improvement starts to slow down. This means that the algorithm successfully uses the structure in the objective functions to find good solutions with little computation time.

Comparing our heuristic approaches with the exact (non-aggregation) method, we find that for few objectives ($K=50$), there is a reasonable trade-off between reduced computation times and loss of performance. For many objectives ($K=200$), it is even possible to find better solutions within the time limit of 60 seconds than when one does not use objective aggregation.

\section{Conclusions}

In this paper, we studied the ordered weighted averaging problem
($\textsc{OWA}~\mathcal{P}$)
 for combinatorial optimization problems~$\mathcal{P}$ with multiple objectives (or cost scenarios in the robust setting), which is a popular approach in decision making literature. The current best approximation algorithm for $\textsc{OWA}~\mathcal{P}$
with
nonincreasing weights from~\cite{KZ16}, under the assumption that~$\mathcal{P}$ with a single objective is polynomially solvable,
 is based on using an 
average-cost objective approach and gives an approximation guarantee of $w_1 K$, where $w_1\in[1/K,1]$ is the largest weight, and $K$ is the number of objectives. 
 By using the aggregation method proposed  ($\ell$-\textsc{Aggregation Algorithm})
that combines the groups of $\ell$ objectives, we were able to improve this bound to a $\ell\rho$ approximation, where $\rho$ depends on the distribution of weights. 

While solving an OWA version of
a combinatorial optimization problem
 with more than one objective is usually not polynomial, it is still possible to derive an $\varepsilon K$-approximation algorithm for any fixed constant $\varepsilon\in (0,1)$ that runs in polynomial time
under assumption that the OWA problem 
has a polynomial 2-approximation algorithm for constant number of objectives.

Furthermore, we examined the Hurwicz criterion, which is a special case of OWA, where only the best case and the worst case of a solution are considered. We showed that it is possible to solve this problem by considering $K$ min-max robust optimization problems, for which some approximation guarantees already exist, thus further improving the best-known approach for this case.

In computational experiments, we tested the practical performance of solutions that are found through the proposed objective aggregation, and another heuristic based on $K$-means aggregation. Results are much closer to optimal objective values than the theoretical bound indicates, and even perform better than the direct solution approach if the number of objectives is large. 

\subsubsection*{Acknowledgements}
Adam Kasperski and Pawe{\l} Zieli{\'n}ski are  supported by the National Science Centre, Poland, grant 2017/25/B/ST6/00486.


\begin{thebibliography}{10}

\bibitem{ABV09}
H.~Aissi, C.~Bazgan, and D.~Vanderpooten.
\newblock Min-max and min-max regret versions of combinatorial optimization
  problems: a survey.
\newblock {\em European Journal of Operational Research}, 197:427--438, 2009.

\bibitem{ABV10}
H.~Aissi, C.~Bazgan, and D.~Vanderpooten.
\newblock General approximation schemes for minmax (regret) versions of some
  (pseudo-)polynomial problems.
\newblock {\em Discrete Optimization}, 7:136--148, 2010.

\bibitem{AFK02}
S.~Arora, A.~Frieze, and H.~Kaplan.
\newblock A new rounding procedure for the assignment problem with applications
  to dense graph arrangement problems.
\newblock {\em Mathematical Programming}, 92:1--36, 2002.

\bibitem{CG15}
A.~B. Chassein and M.~Goerigk.
\newblock Alternative formulations for the ordered weighted averaging
  objective.
\newblock {\em Information Processing Letters}, 115:604--608, 2015.

\bibitem{CG18}
A.~B. Chassein and M.~Goerigk.
\newblock On scenario aggregation to approximate robust combinatorial
  optimization problems.
\newblock {\em Optimization Letters}, DOI: 10.1007/s11590-017-1206-x.

\bibitem{D13}
B.~Doerr.
\newblock Improved approximation algorithms for the {M}in-{M}ax selecting
  {I}tems problem.
\newblock {\em Information Processing Letters}, 113:747--749, 2013.

\bibitem{ER05}
M.~Ehrgott.
\newblock {\em Multicriteria optimization}.
\newblock Springer, 2005.

\bibitem{ER00}
M.~Ehrgott and X.~Gandibleux.
\newblock A survey and annoted bibliography of multiobjective combinatorial
  optimization.
\newblock {\em OR Spectrum}, 22:425--460, 2000.

\bibitem{FPP14}
E.~Fern{\'{a}}ndez, M.~A. Pozo, and J.~Puerto.
\newblock Ordered weighted average combinatorial optimization: {F}ormulations
  and their properties.
\newblock {\em Discrete Applied Mathematics}, 169:97--118, 2014.

\bibitem{FPPS17}
E.~Fern{\'{a}}ndez, M.~A. Pozo, J.~Puerto, and A.~Scozzari.
\newblock {O}rdered {W}eighted {A}verage optimization in {M}ultiobjective
  {S}panning {T}ree {P}roblem.
\newblock {\em European Journal of Operational Research}, 260:886--903, 2017.

\bibitem{GS12}
L.~Galand and O.~Spanjaard.
\newblock Exact algorithms for {OWA}-optimization in multiobjective spanning
  tree problems.
\newblock {\em Computers and Operations Research}, 39:1540--1554, 2012.

\bibitem{jain2010data}
A.~K. Jain.
\newblock Data clustering: 50 years beyond {K}-means.
\newblock {\em Pattern recognition letters}, 31(8):651--666, 2010.

\bibitem{KW94}
P.~Kall and S.~W. Wallace.
\newblock {\em Stochastic Programming}.
\newblock John Wiley and Sons, 1994.

\bibitem{KZ11}
A.~Kasperski and P.~Zieli{\'n}ski.
\newblock On the approximability of robust spanning problems.
\newblock {\em Theoretical Computer Science}, 412:365--374, 2011.

\bibitem{KZ15}
A.~Kasperski and P.~Zieli{\'n}ski.
\newblock Combinatorial optimization problems with uncertain costs and the
  {OWA} criterion.
\newblock {\em Theoretical Computer Science}, 565:102--112, 2015.

\bibitem{KZ16}
A.~Kasperski and P.~Zieli{\'n}ski.
\newblock Using the {WOWA} operator in robust discrete optimization problems.
\newblock {\em International Journal of Approximate Reasoning}, 68:54--67,
  2015.

\bibitem{KZ16b}
A.~Kasperski and P.~Zieli{\'n}ski.
\newblock Robust {D}iscrete {O}ptimization {U}nder {D}iscrete and {I}nterval
  {U}ncertainty: {A} {S}urvey.
\newblock In {\em Robustness {A}nalysis in {D}ecision {A}iding, {O}ptimization,
  and {A}nalytics}, pages 113--143. Springer, 2016.

\bibitem{KMU08}
I.~Katriel, C.~Kenyon-Mathieu, and E.~Upfal.
\newblock Commitment under uncertainty: two-stage matching problems.
\newblock {\em Theoretical Computer Science}, 408:213--223, 2008.

\bibitem{KY97}
P.~Kouvelis and G.~Yu.
\newblock {\em Robust Discrete Optimization and its Applications}.
\newblock Kluwer Academic Publishers, 1997.

\bibitem{LR57}
R.~D. Luce and H.~Raiffa.
\newblock {\em {G}ames and {D}ecisions: {I}ntroduction and {C}ritical
  {S}urvey}.
\newblock Dover Publications Inc., 1989.

\bibitem{M70}
D.~S. Mitrinovi{i\'c}.
\newblock {\em Analytic {I}nequalities}.
\newblock Springer-Verlag, 1970.

\bibitem{OO12}
W.~Ogryczak and P.~Olender.
\newblock On {MILP} models for the {OWA} optimization.
\newblock {\em Journal of Telecommunications and Information Technology},
  2:5--12, 2012.

\bibitem{OS03}
W.~Ogryczak and T.~{\'S}liwi{\'n}ski.
\newblock On solving linear programs with the ordered weighted averaging
  objective.
\newblock {\em European Journal of Operational Research}, 148(1):80--91, 2003.

\bibitem{PS98}
C.~H. Papadimitriou and K.~Steiglitz.
\newblock {\em Combinatorial optimization: algorithms and complexity}.
\newblock Dover Publications Inc., 1998.

\bibitem{RW91}
R.~T. Rockafellar and R.~J.-B. Wets.
\newblock Scenarios and {P}olicy {A}ggregation in {O}ptimization {U}nder
  {U}ncertainty.
\newblock {\em Mathematics of Operations Research}, 16:119--147, 1991.

\bibitem{UT94}
E.~Ulungu and J.~Teghem.
\newblock Multi-objective combinatorial optimization problems: A survey.
\newblock {\em Journal of Multi-criteria Decision Analysis}, pages 83--104,
  1994.

\bibitem{W89}
R.~J.-B. Wets.
\newblock The aggregation principle in scenario analysis and stochastic
  optimization.
\newblock In S.~W. Wallace, editor, {\em Algorithms and Model Formulations in
  Mathematical Programming}, pages 91--113. Springer-Verlag, 1989.

\bibitem{WS10}
D.~P. Williamson and D.~B. Shmoys.
\newblock {\em The {D}esign of {A}pproximation {A}lgorithms}.
\newblock Cambridge University Press, 2010.

\bibitem{YA88}
R.~R. Yager.
\newblock On ordered weighted averaging aggregation operators in multi-criteria
  decision making.
\newblock {\em IEEE Transactions on Systems, Man and Cybernetics}, 18:183--190,
  1988.

\bibitem{YKB11}
R.~R. Yager, J.~Kacprzyk, and G.~Beliakov, editors.
\newblock {\em Recent developments in the {O}rdered {W}eighted {A}veraging
  operators: {T}heory and {P}ractice}.
\newblock Springer, 2011.

\end{thebibliography}
\end{document}